\documentclass{ecai} 

\usepackage{latexsym}
\usepackage{amssymb}
\usepackage{enumitem}
\usepackage{graphicx}
\usepackage{color}
\usepackage{times}
\usepackage{soul}
\usepackage{url}
\usepackage[hidelinks]{hyperref}
\usepackage[utf8]{inputenc}
\usepackage[small]{caption}
\usepackage{graphicx}
\usepackage{amsmath}
\usepackage{amsthm}
\usepackage{booktabs}
\usepackage{algorithm}
\usepackage{algorithmic}
\usepackage[switch]{lineno}

\usepackage{mathtools}
\usepackage{tabularx,lipsum,environ}
\usepackage{booktabs}
\usepackage{multirow}
\usepackage{tikz}
\usepackage{pifont,dsfont} 
\usepackage{nicefrac}
\usepackage{float}
\usepackage{array}
\usepackage{makecell}
\usepackage{longtable}

\DeclarePairedDelimiter\ceil{\lceil}{\rceil}
\DeclarePairedDelimiter\floor{\lfloor}{\rfloor}
\newcommand{\proofonlyif}{\smallskip\textit{Only if:\quad}}
\newcommand{\proofif}{\smallskip\textit{If:\quad}}

\newcommand{\PenroseBanzhaf}{\beta}
\newcommand{\ShapleyShubik}{\varphi}
\newcommand{\PI}{\textsc{PI}}

\newcommand{\p}{\ensuremath{\mathrm{P}}}
\newcommand{\np}{\ensuremath{\mathrm{NP}}}
\newcommand{\conp}{\ensuremath{\mathrm{coNP}}}

\newcommand{\PP}{\ensuremath{\mathrm{PP}}}
\newcommand{\CE}{\ensuremath{\mathrm{C_{=}}}}
\newcommand{\CC}{\ensuremath{\mathrm{C}}}

\newcommand{\OMIT}[1]{}

\newcommand{\littlep}{{\rm p}}
\newcommand{\manyone}{\ensuremath{\mbox{$\,\leq_{\rm m}^{{\littlep}}$\,}}}

\newcommand{\condition}{\,|\:}

\newcommand{\EP}[3]{
	\begin{center}
		\smallskip
		{\small 
			\begin{tabularx}{\columnwidth}{@{}l@{\hspace*{2mm}}l@{}}
				\toprule
				\multicolumn{2}{c}{\sc{#1}} \\
				\midrule
				{\bf Given:}& \parbox[t]{0.80\columnwidth}{#2\vspace*{1mm}} \\
				{\bf Question:}& \parbox[t]{0.80\columnwidth}{#3\vspace*{.5mm}} \\ 
				\bottomrule
			\end{tabularx}
		}
		\smallskip
	\end{center}
}

\newenvironment{proofsketch}{\noindent{\textit{Proof Sketch.}}}{\literalqed\bigskip}

\def\literalqed{{\ \nolinebreak\hfill\mbox{\qed\quad}}}

\newcommand{\sproofof}[1]{\noindent{\textit{Proof of~{#1}.}}}
\newcommand{\eproofof}[1]{\noindent{\hspace*{0.1in} \hfill \literalqed~{\scriptsize #1}}\bigskip}

\newtheorem{theorem}{Theorem}
\newtheorem{lemma}[theorem]{Lemma}

\newtheorem{definition}{Definition}

\newcommand{\BibTeX}{B\kern-.05em{\sc i\kern-.025em b}\kern-.08em\TeX}

\begin{document}

\begin{frontmatter}

\paperid{1869}

\title{Control by Adding Players to Change or Maintain the Shapley--Shubik or
	the Penrose--Banzhaf Power Index in Weighted Voting Games Is Complete
	for NP$^{\textbf{PP}}$}

\author[A]{\fnms{Joanna}~\snm{Kaczmarek}\orcid{0000-0001-6652-6433}
\thanks{Corresponding Author. Email: joanna.kaczmarek@hhu.de}}
\author[A]{\fnms{J\"{o}rg}~\snm{Rothe}\orcid{0000-0002-0589-3616}}

\address[A]{Institut f\"{u}r Informatik, MNF, Heinrich-Heine-Universität D\"{u}sseldorf, D\"{u}sseldorf, Germany}

\begin{abstract}
Weighted voting games are a well-known and useful class of succinctly representable simple games that have many real-world applications, e.g., to model collective decision-making in legislative bodies or shareholder voting.
Among the structural control types being analyzing, one is control by adding players to weighted voting games, so as to either change or to maintain a player's power in the sense of the (probabilistic) Penrose--Banzhaf power index
or the Shapley--Shubik power index.
For the
problems related to this control, the best known lower bound is $\PP$-hardness, where $\PP$ is ``probabilistic polynomial time,'' and the best known upper bound is the class $\np^{\PP}$, i.e., the class $\np$ with a $\PP$ oracle.
We optimally raise this lower bound by showing $\np^{\PP}$-hardness of all 
these
problems 
for 
the Penrose--Banzhaf and the Shapley--Shubik 
indices, 
thus establishing completeness for them in that class.
Our proof technique may turn out to be useful for solving other open problems related to weighted voting games with such a complexity gap as well.
\end{abstract}

\end{frontmatter}

\section{Introduction}

Weighted voting games (WVGs) are a central, very popular class of simple coalitional games with many real-world applications.
They can be used to model and analyze collective decision-making in legislative bodies and in parliamentary voting~\cite{sha:b:polsci:power}, such as the European Union or the International Monetary Fund~\cite{gam:j:power-indices-for-political-and-financial-decision-making-review}, in joint stock companies, etc.
For more information, we refer to the books by Chalkiadakis \emph{et al.}~\citep{cha-elk-woo:b:computational-aspects-of-cooperative-game-theory}, Taylor and Zwicker~\citep{tay-zwi:b:simple-games}, and Peleg and Sudh{\"{o}}lter~\citep{pel-sud:b:cooperative-games} and the book chapters by 
Chalkiadakis and Wooldridge~\citep{cha-woo:b:handbook-comsoc-weighted-voting-games} and 
Bullinger \emph{et al.}~\citep{bul-elk-rot:b-2nd-edition:economics-and-computation-cooperative-game-theory}.
Especially important is the analysis of how significant players are in 
WVGs,
i.e., what they contribute to forming winning coalitions.
Their influence
can be measured by so-called power indices among which 
some well-known examples are:
the \emph{Shapley--Shubik index} due to
Shapley and Shubik~\citep{sha-shu:j:shapley-shubik-index},
the \emph{probabilistic Penrose--Banzhaf index} due to
Dubey and Shapley~\citep{dub-sha:j:banzhaf}, 
and also the \emph{normalized Penrose--Banzhaf
	index} due to Penrose~\citep{pen:j:banzhaf-index} and
Banzhaf~\citep{ban:j:weighted-voting-doesnt-work}.
We are concerned with the former two.

Much work has been done on how one can tamper with a given player's power in a
WVG.
For example, the effect of merging or splitting players (the latter a.k.a.\ ``false-name manipulation'')  was studied by Aziz \emph{et al.}~\citep{azi-bac-elk-pat:j:false-name-manipulations-wvg} and later on by Rey and Rothe~\citep{rey-rot:j:false-name-manipulation-PP-hard}.
Zuckerman \emph{et al.}~\citep{zuc-fal-bac-elk:j:manipulating-quote-in-wvg} studied the impact of manipulating the quota in 
WVGs
on the power of players.
Another way of tampering with the players' power was introduced by Rey and Rothe~\citep{rey-rot:j:structural-control-in-weighted-voting-games} who studied control problems by adding players to or by deleting players from a 
WVG;
their results have recently been improved 
by Kaczmarek and Rothe~\citep{kac-rot:j:controlling-weighted-voting-games-by-deleting-or-adding-players-with-or-without-changing-the-quota}.

Control attempts in voting (e.g., by adding or deleting either voters or
candidates) have been studied in
depth~\cite{fal-rot:b:handbook-comsoc-control-and-bribery}.
Surprisingly, however, much less work has been done on control
attempts in cooperative game theory, such as for 
WVGs
(e.g., by adding or deleting players).
Control by adding players to
WVGs is inspired by the analogous notion of control by adding either
candidates or voters to elections in voting.  There are
many real-world scenarios where WVGs and power indices are used to
analyze the power of agents and where there is an incentive to change
the power in the situation to somebody's advantage (e.g., in politics
or to measure control in corporate structures).  Concretely, WVGs are
the typical way to model decision-making in the EU, as countries can
be assigned a weight 
(essentially related to their population size).  
The EU is
constantly expanding: New members join in (or, rarely, they leave),
which is exactly
control by adding players,
raising the question of if and how the power of old EU members is
changed by adding new ones to the EU---just one clear-cut case of
motivation among various others.  If new members join, an old one may
insist on having the same power afterwards (motivating the goal of
``maintaining one's power''), or at least not lose power
(``nondecreasing one's power''), or Poland may insist that Germany's
power does not increase when Ukraine joins (``nonincreasing one's
power'').
We continue the work on the computational complexity of structural control by adding players to a weighted voting game started by Rey and Rothe~\citep{rey-rot:j:structural-control-in-weighted-voting-games}.
They showed $\PP$-hardness for the related problems and an upper bound of $\np^{\PP}$.  We optimally improve their results by showing $\np^{\PP}$-completeness for these problems.

Many of the problems related to 
WVGs
are computationally hard.
For instance, under suitable functional reducibilities, computing the Shapley--Shubik power index~\cite{den-pap:j:comparative-solution-concepts} and the Penrose--Banzhaf power indices~\cite{pra-kel:j:voting} is $\#\p$-complete, where $\#\p$ is the \emph{counting version of the class~$\np$}~\cite{val:j:permanent}.
This is employed by Faliszewski and Hemaspaandra~\cite{fal-hem:j:power-index-comparison} in their result that comparing a given player's probabilistic Penrose--Banzhaf index or a given player's Shapley--Shubik index in two given 
WVGs
is $\PP$-complete.
$\PP$ is \emph{probabilistic polynomial time}~\cite{gil:j:probabilistic-tms}, a complexity class that is presumably
larger than the class~$\np$.

Adding players is just one possibility to change the outcome of a WVG; as mentioned above, Aziz \emph{et al.}~\citep{azi-bac-elk-pat:j:false-name-manipulations-wvg} proposed merging or splitting players so as to change their power.
The problems related to merging players in 
WVGs
were later proven to be $\PP$-complete~\cite{rey-rot:j:false-name-manipulation-PP-hard}.
However, interestingly, the same complexity gap we are concerned with here---$\PP$-hardness versus membership in $\np^{\PP}$---is also persistent for false-name manipulation, i.e., for the problems related to splitting players~\cite{rey-rot:j:false-name-manipulation-PP-hard}.
The novel proof techniques developed in the current paper may thus turn out to be useful for closing this 
huge
complexity gap as well, which provides another strong motivation of our work.
There are many interesting open problems in the literature on WVGs---another one is control by adding or deleting edges in graph-restricted
WVGs, again with a complexity gap between $\PP$-hardness and membership in $\np^{\PP}$~\cite{kac-rot-tal:c:complexity-of-control-by-adding-or-deleting-edges-in-graph-restricted-weighted-voting-games}---and our novel
approach might be useful to settle them as well.

We start with providing the needed notions from cooperative game theory and computational complexity in Section~\ref{sec:preliminaries},
and introduce a new $\np^\PP$-complete problem which is used in some of our reductions. 
In Section~\ref{sec:prereduction}, we prepare some tools and show their properties that are needed in our proofs.
Finally in Section~\ref{sec:control-problems}, we present our results.
Due to space limitations, some of our proofs are moved to the technical appendix.

\section{Preliminaries}
\label{sec:preliminaries}

We start by
recalling some notions from cooperative game theory.
Let $N = \{1, \ldots , n\}$ be a set of players.
For $v : 2^N \rightarrow \mathbb{R}_{\ge 0}$, where
$\mathbb{R}_{\ge 0}$ denotes the set of nonnegative real numbers, a \emph{coalitional game} is a pair $(N,v)$ and each subset of $N$ is called a \emph{coalition}.
$(N,v)$ is a \emph{simple} coalitional game if it is \emph{monotonic} (i.e., $v(T) \leq v(T')$ for any $T,T'$ with $T \subseteq T' \subseteq N$), and $v(S) \in \{0,1\}$ for each coalition $S \subseteq N$.
We focus on the following type of simple coalitional games.

\begin{definition}
	A \emph{weighted voting game}
	$\mathcal{G} = (w_{1}, \ldots, w_{n};q)$
	is a simple coalitional game with player set $N$ that
	consists of a natural number
	$q$ called the \emph{quota} and nonnegative integer
	weights, where $w_{i}$ is the \emph{weight of player~$i \in N$}.
	For each coalition $S \subseteq N$, let  $w_{S} = \sum_{i \in S} w_{i}$
	and define the \emph{characteristic function
		$v : 2^N \rightarrow \{0,1\}$ of $\mathcal{G}$} as
	$v(S) = 1$ if $w_{S} \ge q$, and $v(S) = 0$ otherwise.
	We say that $S$ is a \emph{winning coalition} if $v(S)=1$,
	and it is a \emph{losing coalition} if $v(S)=0$.
	Moreover, we call a player~$i$ \emph{pivotal for coalition
		$S\subseteq N\setminus\{i\}$} if $v(S\cup \{i\})-v(S)=1$. 
\end{definition}

One of the things we want to know about players is how significant they are in a given game.
We usually measure this by so-called \emph{power indices}.
The main information used in determining the power index of a player~$i$ is the number of coalitions $i$ is pivotal for.
We study two of the most popular and well-known power indices.
One of them is the \emph{probabilistic Penrose--Banzhaf power index},
which was introduced by Dubey and Shapley~\citep{dub-sha:j:banzhaf} as an
alternative to the original \emph{normalized Penrose--Banzhaf
	index}~\cite{pen:j:banzhaf-index,ban:j:weighted-voting-doesnt-work}.

\begin{definition}\label{PBI}
	Let $\mathcal{G}$ be a WVG.
	The \emph{probabilistic Penrose--Banzhaf power index of a player~$i$ in
		$\mathcal{G}$} is defined by
	\[
	\PenroseBanzhaf(\mathcal{G},i) =
	\frac{1}{2^{n-1}}\sum\limits_{S \subseteq N \setminus \{i\}}(v(S \cup \{i\})-v(S)).
	\]
\end{definition}

The other index we will study is the \emph{Shapley--Shubik power index}, introduced by Shapley and Shubik~\citep{sha-shu:j:shapley-shubik-index} as follows:

\begin{definition}\label{SSI}
	Let $\mathcal{G}$ be a WVG.
	The \emph{Shapley--Shubik power index of a player~$i$ in $\mathcal{G}$}
	is defined by
	\[
	\ShapleyShubik(\mathcal{G},i) =
	\frac{1}{n!}\sum\limits_{S \subseteq N \setminus \{i\}}|S|!(n-1-|S|)!(v(S \cup \{i\})-v(S)).
	\]
\end{definition}

We assume familiarity with the basic concepts of computational
complexity theory, such as the well-known complexity classes $\p$
(\emph{deterministic polynomial time}), $\np$ (\emph{nondeterministic
	polynomial time}), and $\PP$ (\emph{probabilistic polynomial
	time}~\cite{gil:j:probabilistic-tms}).
$\np^{\PP}$ is the class of problems that can be solved by an $\np$
oracle Turing machine accessing a $\PP$ oracle.
It is a very large complexity class containing even
the entire polynomial hierarchy by Toda's result~\cite{tod:j:pp-ph}.

We will use the notions of
completeness and hardness for a complexity class based on the
polynomial-time many-one reducibility: A problem $X$ \emph{(polynomial-time
	many-one) reduces to a problem $Y$} ($X \manyone Y$) if
there is a polynomial-time computable function $\rho$ such that
for each input~$x$, $x \in X \Longleftrightarrow \rho(x) \in Y$; $Y$ is
hard for a complexity class $\mathcal{C}$ if $C \manyone Y$ for each
$C \in \mathcal{C}$; 
and $Y$ is complete for $\mathcal{C}$ if $Y$ is
$\mathcal{C}$-hard and $Y \in \mathcal{C}$.  For more background on
complexity theory, we refer to some of the common text
books~\cite{gar-joh:b:int,pap:b:complexity,rot:b:cryptocomplexity}.

Valiant~\citep{val:j:permanent} introduced $\#\p$ as the class of
functions that give the number of solutions of $\np$ problems.
$\#\p$ is a.k.a.\ the \emph{``counting version of $~\np$''}:
For every $\np$ problem $X$, $\#X$ denotes the function that maps each
instance of $X$ to the number of its solutions.  For example, for
the problem
$\textsc{SAT} = \{\phi \condition \phi$ is a boolean
formula satisfied by at least one truth assignment$\}$,
which is $\np$-complete~\cite{coo:c:theorem-proving},
$\#\textsc{SAT}$ maps each boolean formula to the number of its
satisfying assignments.  
Clearly, any $\np$ problem $X$ is
closely related to its counting version $\#X$ because if we can efficiently
count the number of solutions of an instance~$x$, we can immediately
tell whether $x$ is a yes- or a no-instance of~$X$: $x \in X$ exactly
if the number of solutions of $x$ is positive.

Deng and Papadimitriou~\citep{den-pap:j:comparative-solution-concepts}
showed that computing the Shapley--Shubik index of a player in a given
WVG is complete for~$\#\p$ via \emph{functional}
many-one reductions.
Prasad and Kelly~\citep{pra-kel:j:voting} proved that computing the
probabilistic Penrose--Banzhaf index is parsimoniously complete
for~$\#\p$.  
$\#\p$ and $\PP$, even though the former is a class of functions and
the latter a class of decision problems, are closely related by the
well-known result that $\p^\PP = \p^{\#\p}$.
For more complexity-theoretic background on the
\emph{counting (polynomial-time) hierarchy}, which contains~$\np^\PP$,
we refer to \cite{wag:j:succinct,par-sch:j:parallel-computation-with-threshold-functions,tor:j:quantifiers,tod:j:pp-ph,rot:b:cryptocomplexity}. 
Using the standard problem complete for $\PP$ due to
Gill~\citep{gil:j:probabilistic-tms}, i.e., $\textsc{MajSAT}
= \{\phi \condition \phi$ is a boolean formula satisfied by a majority
of truth assignments$\}$,
Littman \emph{et al.}~\citep{lit-gol-mun:j:complexity-probabilistic-planning}
introduced and studied the following problem that they proved to be
$\np^\PP$-complete:

\EP{\textsc{Exist-Majority-SAT} (\textsc{E-MajSAT})}
{A boolean formula $\phi$ with $n$ variables $x_1,\ldots,x_n$ and an
	integer~$k$, $1\le k\le n$.}
{Is there an assignment to the first $k$ variables $x_1,\ldots,x_k$
	such that a majority of assignments to the remaining $n-k$
	variables $x_{k+1},\ldots,x_n$ satisfies~$\phi$?}

Another 
closely related $\np^\PP$-complete decision problem was
introduced by de Campos et
al.~\citep{cam-sta-wey:c:complexity-stochastic-optimization}:

\EP{\textsc{Exist-Minority-SAT} (\textsc{E-MinSAT})}
{A boolean formula $\phi$ with $n$ variables $x_1,\ldots,x_n$ and an
	integer~$k$, $1\le k\le n$.}
{Is there an assignment to the first $k$ variables $x_1,\ldots,x_k$
	such that at most half of the assignments to the remaining $n-k$
	variables $x_{k+1},\ldots,x_n$ satisfies $\phi$?}

Note that if $k=0$, \textsc{E-MajSAT} is equivalent to the
$\PP$-complete problem \textsc{MajSAT}, and \textsc{E-MinSAT} is
equivalent to the complement of \textsc{MajSAT}, which is also
$\PP$-complete since the class $\PP$ is closed under
complement~\cite{gil:j:probabilistic-tms}.  If $k=n$,
\textsc{E-MajSAT} is equivalent to the $\np$-complete problem
\textsc{SAT}, and \textsc{E-MinSAT} is equivalent to the
complement of \textsc{SAT}, i.e., it is $\conp$-complete.
Therefore, we can omit these cases ($k=0$ and $k=n$) when proving
$\np^{\PP}$-hardness of our problems.  Moreover, we can also assume
that a given formula in CNF does not contain any variable $x$ in both
forms, $x$ and $\neg x$, in any of its clause (which can be checked in
polynomial time) 
because then the clause would be true for any possible
truth assignment.  Also, we will assume that our inputs for these
problems contain only those variables that actually occur in the
given boolean formula.

Rey and Rothe~\citep{rey-rot:j:structural-control-in-weighted-voting-games} defined problems capturing
control by adding players to a given WVG so as to change a given
player's power in the modified game.
To increase this power for an index \textsc{PI}, 
the control
problem is defined as follows:

\EP{$\textsc{Control-by-Adding-Players-to-Increase-PI}$}
{A WVG $\mathcal{G}$ with a set $N$ of players, a set $M$ of players (given by their weights) that can be added to~$\mathcal{G}$,
	a distinguished player~$p \in N$, and a positive integer
	$k\le \|M\|$.}
{Can at most $k$ players $M' \subseteq M$
	be added to $\mathcal{G}$ such that for the new game
	$\mathcal{G}_{\cup M'}$, it holds that
	$\PI(\mathcal{G}_{\cup M'},p) > \PI(\mathcal{G},p)$?}
The corresponding control problems for decreasing, nonincreasing, nondecreasing, and maintaining \textsc{PI} are defined analogously, by changing the relation sign in the question to ``$<$,'' ``$\le$,'' ``$\ge$,'' and ``$=$,''
respectively.
Additionally, we assume that we add at least one new player 
in case of nondecreasing, nonincreasing, or maintaining \textsc{PI}
(otherwise, the control problems would be trivial).

For both the Penrose--Banzhaf and the Shapley--Shubik power index, Rey and Rothe~\citep{rey-rot:j:structural-control-in-weighted-voting-games} showed that these five control problems are $\PP$-hard, and they observed that $\np^{\PP}$ is the best known upper bound for them.
Our goal in this paper is to raise the $\PP$-hardness lower bound of these problems to $\np^{\PP}$-hardness, thus establishing their completeness 
for this class.
We now introduce another problem that will be used in some of our
proofs and state its $\np^{\PP}$-completeness:

\EP{\textsc{Exist-Exact-SAT} (\textsc{E-ExaSAT})}
{A boolean formula $\phi$ with $n$ variables $x_1,\ldots,x_n$, an
	integer~$k$, $1\le k\le n$, and an integer~$\ell$.}
{Is there an assignment to the first $k$ variables $x_1,\ldots,x_k$
	such that \emph{exactly} $\ell$ assignments to the remaining
	$n-k$ variables $x_{k+1},\ldots,x_n$ 
	satisfy~$\phi$?}

\begin{lemma}\label{lem:e-exact-sat}
	\label{lem:E-Exact-SAT}
	\textsc{E-ExaSAT} is $\np^{\PP}$-complete.
\end{lemma}

The proof of Lemma~\ref{lem:e-exact-sat} can be found in the technical
appendix.

\OMIT{
	Finally, note that for $\ell=0$, \textsc{E-ExaSAT} is in $\np^{\conp}$;
	therefore, we can omit this case when proving $\np^{\PP}$-completeness.
} %

\section{%
	Transforming Value Assignments of Boolean Formulas to Weight Vectors}
\label{sec:prereduction}

First, let us define a transformation from a value assignment for a given boolean formula to vectors of weights to be used for some players in our reductions later on.

\begin{definition}
	\label{def:prereduction}
	Let $\phi$ be given boolean formula in CNF
	with variables $x_1,\dots,x_n$ and $m$ clauses.
	Let $k\in\mathbb{N}$ with $k\le n$ and $r = \ceil{\log_{2} n}-1$.
	Let us define the following two sets of weight vectors which are going to be assigned as weights to players divided either into three sets---$M$, $A$, and $C$---or into four sets---$M$, $A$, $C$, and $C'$---in our proofs later on:
	\begin{description}
		\item[Set 1:] For some $t \in\mathbb{N}\setminus\{0\}$ such that
		$10^{t} > 2^{\ceil{\log_2 n}+1}$, and
		for $i \in \{1,\ldots,n\}$, define
		\begin{eqnarray*}
			a_i & = & 10^{t(m+1)+i}+\sum_{\substack{j\,:\ \textrm{clause $j$} \\ \textrm{ contains $x_i$}}} 10^{tj}
			\text{ and } \\
			b_i & = & 10^{t(m+1)+i}+\sum_{\substack{j\,:\ \textrm{clause $j$} \\ \textrm{ contains $\neg x_i$}}} 10^{tj},
		\end{eqnarray*}
		and for $j\in\{1,\ldots,m\}$ and $s\in\{0,\ldots,r\}$, define
		\begin{eqnarray*}
			c_{j,s} & = & 2^{s}\cdot 10^{tj}.
		\end{eqnarray*}
		Define the following three weight vectors:
		\begin{eqnarray*}
			W_M & = & (a_1,\ldots,a_k,b_1,\ldots,b_k), \\
			W_A & = & (a_{k+1},\ldots,a_n,b_{k+1},\ldots,b_n),\\
			W_C & = & (c_{1,0},\ldots,c_{m,r}).
		\end{eqnarray*}
		
		\item[Set 2:] For some $t,t' \in\mathbb{N}\setminus\{0\}$ such that
		$10^{t'} > 2^{\ceil{\log_2 n}+1}$ 
		and 
		$10^t> 10^{t'} + 2^{\ceil{\log_2 n}+1}\sum_{l=1}^{m}10^{lt'}$, 
		and
		for $i \in \{1,\ldots,n\}$, define $a_i$ and $b_i$
		as in Set~1,
		\OMIT{
			\begin{eqnarray*}
				a_i & = & 10^{t(m+1)+i}+\sum_{\substack{j\,:\ \textrm{clause $j$} \\ \textrm{ contains $x_i$}}} 10^{tj}
				\text{ and } \\
				b_i & = & 10^{t(m+1)+i}+\sum_{\substack{j\,:\ \textrm{clause $j$} \\ \textrm{ contains $\neg x_i$}}} 10^{tj},
			\end{eqnarray*}
		} %
		and for $j\in\{1,\ldots,m\}$ and $s\in\{0,\ldots,r\}$, let
		\[
		c_{j,s}'  =  2^{s}\cdot 10^{t'j}
		\text{ and } 
		c_{j,s}  =  2^{s}\cdot 10^{tj}+c_{j,s}'.
		\]
		In addition to $W_M$ and $W_A$ defined as in Set~1,
		define the following two weight vectors: 
		\[
		W_{C'}  =  (c_{1,0}',\ldots,c_{m,r}')
		\text{ and }
		W_C  =  (c_{1,0},\ldots,c_{m,r}).
		\]
	\end{description}	
	Additionally, let
	\begin{align*}
		q_{1}  ~=~ & \sum_{i=1}^{n} 10^{t(m+1)+i} + 2^{\ceil{\log_{2} n}}\sum_{j=1}^{m} 10^{tj} \textrm{ and }\\
		q_{2}  ~=~ & \sum_{i=1}^{n} 10^{t(m+1)+i} + 2^{\ceil{\log_{2} n}}\sum_{j=1}^{m} 10^{tj}\\
		& + \left(2^{\ceil{\log_{2} n}}-1 \right)\sum_{j=1}^{m} 10^{t'j}.
	\end{align*}
\end{definition}

\begin{lemma}
	\label{lem:correspondence-assignments-weights}
	Let $i\in\{1,2\}$. There exists a bijective transformation from the set of value assignments satisfying a boolean formula $\phi$ to the family of subsets of players with weights defined in Set~$i$ of Definition~\ref{def:prereduction} whose total weight equals $q_{i}$.  
\end{lemma}
\begin{proofsketch}
	It can be shown that for each set $S$ of weight~$q_i$, for $i\in\{1,2\}$, $S$ has to
	contain exactly $n$ players from $M\cup A$ (namely, $n$ players, each
	with exactly one weight from $\{a_j,b_j\}$, $j\in\{1,\ldots,n\}$), and
	for each $S\cap (M\cup A)$, there exists exactly one set of
	weight~$q_1$ with players from $C$ for Set~$1$ and~$q_2$ from $C\cup C'$ for Set~$2$ (but there can exist subsets of $M\cup A$ of the
	mentioned form that are not contained in any set of weight~$q_i$). 
	We present the details in the technical appendix.
	
	Let us prove that there exists a bijection between the sets of
	weight~$q_i$ and the set of value assignments to the variables
	$x_1,\ldots,x_n$ satisfying the given formula~$\phi$.
	
	For each value
	assignment to the variables $x_1,\ldots,x_n$, let
	$1$ represent \texttt{true} and $0$ \texttt{false}, and let
	\begin{equation}
		d_l = \left\{
		\begin{array}{cc}
			a_l & \textrm{ if } x_l = 1, \\
			b_l & \textrm{ if } x_l = 0.
		\end{array}
		\right. 
	\end{equation}
	The resulting weight vector $\vec{d} = (d_1,\ldots,d_n)$ is unique for
	each assignment to $x_1,\ldots,x_n$ (from the previously mentioned
	assumption that no clause contains both a variable and its negation,
	so $a_l \neq b_l$ for any $l\in\{1,\ldots,n\}$).  Also, if this
	vector~$\vec{d}$ corresponds to a satisfying assignment of~$\phi$, the
	total weight of the players' subset in both cases of Set~1 and Set~2 equals
	\[
	\sum_{l=1}^{n}d_l = \sum_{l=1}^{n} 10^{t(m+1)+l} + \sum_{j=1}^{m}p_{j} 10^{tj},
	\]
	where $p_j$, $1\le p_j \le n$, is at least $1$ since each clause is
	satisfied by our fixed assignment: For each clause~$j$, there exists
	some $x_l$ making it true (i.e., either $x_l=1$ and the clause~$j$
	contains~$x_l$, or $x_l=0$ and~$j$ contains~$\neg x_l$), which implies
	that the corresponding $d_l$ has $10^{tj}$ as one of its summands
	(i.e., either $d_l=a_l$ if $x_l$ is contained in clause~$j$, or
	$d_l=b_l$ if $\neg x_l$ is contained in~$j$).  
	From the fact that
	$p_j\neq 0$ for
	all $j\in\{1,\ldots,m\}$~%
	and 
	the previous analysis, there exists
	exactly one subset of 
	$C$ when $i=1$ or exactly one subset of
	$C\cup C'$ when $i=2$ such that the players with the
	corresponding weights together with the players whose weights
	correspond to $\vec{d}$ form a coalition of weight~$q_i$. Therefore,
	for each value assignment satisfying~$\phi$, there exists a unique set
	of players from
	$A\cup M \cup C$ (respectively, $A\cup M \cup C \cup C'$)
	with total weight~$q_i$.
	
	Conversely, let $S\subseteq M \cup A \cup C$ for $i=1$, and
	$S\subseteq M \cup A \cup C \cup C'$ for $i=2$, be a coalition
	of players whose total weight is~$q_i$.  From the previous analysis,
	$S$ can contain exactly one player with weight from $\{a_j, b_j\}$ for
	$j\in\{1,\ldots,n\}$, and for $S\cap (M\cup A)$, there exists exactly
	one subset of $C$ for $i=1$, and exactly one subset of $C\cup
	C'$ for $i=2$, which creates with the former a coalition of
	players with total weight~$q_i$, i.e., there exist no two different
	sets $S$ and $S'$ both with $w_S = w_{S'} = q_i$ such that
	$S\cap(M\cup A) = S' \cap (M\cup A)$.
	
	For the set $S\cap (M\cup A)$ with the weight vector
	$(d_1,\ldots,d_n)$, set
	\begin{equation}\label{addingplayers-increasePBI:def:x_i}
		x_\ell = \left\{
		\begin{array}{ll}
			1 & \textrm{ if } d_\ell = a_\ell \\
			0 & \textrm{ if } d_\ell = b_\ell
		\end{array}
		\right. 
	\end{equation}
	for $\ell\in\{1,\ldots,n\}$.
	For each clause $j\in \{1,\ldots,m\}$,
	there exists some $d_\ell$ corresponding to the 
	player whose weight's 
	part is equal to $10^{tj}$; and if the
	weight
	is~$a_\ell$, clause $j$ contains~$x_\ell$, so assigning
	\texttt{true} to $x_\ell$ makes clause $j$ true; otherwise, the 
	player's weight
	is $b_\ell$ and the clause $j$ contains~$\neg x_\ell$, so assigning
	\texttt{false} to $x_\ell$ makes $j$ true.  Hence, this is a unique value
	assignment to the variables $x_1,\ldots,x_n$ that satisfies~$\phi$ and
	is obtained by the described transformation from the set~$S$.~\end{proofsketch}

The full proof of Lemma~\ref{lem:correspondence-assignments-weights}
can be found in the technical appendix.

\section{NP$^{\text{PP}}$-Hardness of Control by Adding Players to a Weighted
	Voting Game}
\label{sec:control-problems}

In this section, we show our results, i.e., we prove $\np^{\PP}$-hardness of the control problems by adding players to a given WVG.
Specifically, we will present full proofs of $\np^{\PP}$-hardness for
three of the problems.
The remaining proofs
(see Theorem~\ref{restofresults}) can be found in the
appendix.

\begin{theorem}\label{adding-PBI-increase}
	\textsc{Control-by-Adding-Players-to-In\-crease-$\PenroseBanzhaf$}
	is $\np^{\PP}$-complete.
\end{theorem}
\begin{proof}
	We will prove $\np^{\PP}$-hardness by using a reduction from
	\textsc{E-MajSAT}.  Let $(\phi, k)$ be a given instance of
	\textsc{E-MajSAT}, where $\phi$ is a boolean formula in CNF
	with variables $x_1,\dots,x_n$ and $m$ clauses,
	and   
	$1\le k<n$.  
	Before we construct
	an instance of our control problem from $(\phi, k)$, we need to
	choose some numbers and introduce some notation.
	
	Let $t \in \mathbb{N}$ be such that 
	\begin{equation}\label{addingplayers-increasePBI:def:t}
		10^t > \max\left\{2^{\ceil{\log_{2} n}+1}, k+(n-k-1)(k+1)\right\},
	\end{equation}
	and for that $t$, given $\phi$ and~$k$, we define $q_1$ and $W_A$, $W_C$, and $W_M$ as in Set~1 of Definition~\ref{def:prereduction} for player sets~$A$, $C$, and~$M$.

	Now, we construct an instance of
	\textsc{Control-by-Adding-Players-to-Increase~$\PenroseBanzhaf$}:
	Let $k$ be the limit for the number of players that can be added, and
	let $M$ be the set of $2k$ players that can be added with the list of
	weights~$W_M$.
	Further, we define the quota of the WVG $\mathcal{G}$ by
	\begin{equation}\label{addingplayers-increasePBI:def:q}
		q=2\cdot \left(w_A + w_M + w_C + (n-k)(k+1)\right) +1,
	\end{equation}
	and we let $N$ be the set of $4n-2k+m(r+1)$ players in $\mathcal{G}$,
	subdivided into the following seven groups:
	\begin{itemize}
		\item player~$p$ with weight $1$ will be our
		distinguished player,
		\item group $A$ contains $2(n-k)$ players with weight list~$W_A$,
		\item group $C$ contains $m(r+1)$ players with weight list~$W_C$,
		\item group $W$ contains $k$ players with weight list
		\[
		(q-q_1-2, q-q_1-3, \ldots, q-q_1-(k+1)),
		\]
		\item group $X$ contains $k$ players with weight~$1$ each,
		\item group 
		$Y$ contains $n-k$ players with weight list
		\[
		(q-1, q-1 -(k+1), \ldots, q-1 -(n-k-1)(k+1)),\quad\text{ and}
		\]
		\item group $Z$ contains $n-k-1$ players with weight~$k+1$ each.
	\end{itemize}
	
	This concludes the description of how to construct the instance
	$(\mathcal{G}, M, p, k)$ of our control problem from the given
	instance $(\phi, k)$ of \textsc{E-MajSAT}.
	Obviously, this can be done in polynomial time.
	
	Let us first discuss which coalitions player~$p$ can be pivotal for
	in any of the games $\mathcal{G}_{\cup M'}$ for some
	$M' \subseteq M$.\footnote{This also includes the case of the
		unchanged game~$\mathcal{G}$ itself, namely for
		$M' = \emptyset$.}
	Player~$p$ is pivotal for those coalitions of players in
	$(N\setminus \{p\}) \cup M'$ whose total weight is $q-1$.
	First, note that any two players from $W\cup Y$
	together have a weight larger than~$q$.
	Therefore, at most one player from $W\cup Y$
	can be in any coalition player~$p$ is pivotal for.  Moreover,
	by~(\ref{addingplayers-increasePBI:def:q}), all players from $A 
	\cup C
	\cup M \cup X \cup Z$
	together have a total weight smaller than $q-1$.
	This means that any coalition $S \subseteq (N\setminus \{p\}) \cup
	M'$ with a total weight of $q-1$ has to contain \emph{exactly} one of
	the players in $W\cup Y$.  Now, whether this player is in $W$ or 
	$Y$
	has consequences as to which other players will also be in such a
	weight-$(q-1)$ coalition~$S$:
	
	\begin{description}
		\item[Case 1:] If $S$ contains a player from~$W$ with weight, say,
		$q-q_1-\ell-1$ for some~$\ell$, $1 \leq \ell \leq k$, $S$ also has to
		contain those players from $A\cup 
		C\cup
		M$ whose weights sum up 
		to $q_1$
		and $j$ players from~$X$. 
		Indeed, $w_{X\cup Z}< 10^t$, so players from $A\cup C\cup M$
		are needed to achieve $q_1+\ell$. Moreover, they are able to achieve only the value $q_1$ because any subset of  $A\cup C\cup M$ is divisible by $10^t$. At the same time, each player in $Z$ has weight $k+1>\ell$, so no coalition with them achieves $q_1+\ell$. 
		Also, recall that    
		$q_1$ can be achieved only by a set
		of players whose weights take exactly one of the values from
		$\{a_i,b_i\}$ for each $i\in\{1,\ldots,n\}$,
		so $S$ must contain exactly $n-k$ players from  
		$A$ that
		already are in~$\mathcal{G}$ (either $a_i$ or $b_i$, for $k+1 \leq
		i \leq n$) and exactly $k$ players from $M$ (either $a_i$ or
		$b_i$, for $1 \leq i \leq k$); these $k$ players must have been
		added to the game, i.e., $\|M'\|=k$.
		
		\item[Case 2:] If $S$ contains a player from 
		$Y$ with weight, say, $q-1-\ell(k+1)$
		for some~$\ell$, $0 \leq \ell \leq n-k-1$, then either $S$ already
		achieves weight $q-1$ for $\ell=0$, or $S$ has to contain $\ell>0$
		players from~$Z$. The players
		from $X$ are not heavy enough and since each player from
		$A\cup C\cup M$ has a weight larger than $w_{X\cup Z}$    
		(which, together with any player from~$S$, gives a total weight
		exceeding the quota).
	\end{description}

	Since there are no players with weights $a_i$ or $b_i$ for
	$i\in\{1,\ldots,k\}$ in the game $\mathcal{G}$, player~$p$ can be
	pivotal only for the coalitions described in the second case above,
	and therefore,
	\[
	\PenroseBanzhaf(\mathcal{G},p)=\frac{\sum_{j=0}^{n-k-1}{n-k-1 \choose j}}{2^{\|N\|-1}}=\frac{2^{n-k-1}}{2^{\|N\|-1}}.
	\]
	
	We now show the correctness of our reduction: $(\phi, k)$ is a
	yes-instance of \textsc{E-MajSAT} if and only if
	$(\mathcal{G}, M, p, k)$ as defined above is a yes-instance of
	\textsc{Control-by-Adding-Players-to-Increase-$\PenroseBanzhaf$}.

	\proofonlyif
	Suppose that $(\phi, k)$ is a yes-instance of \textsc{E-MajSAT},
	i.e., there exists an assignment to $x_1,\ldots,x_k$ such that a
	majority of assignments to the remaining $n-k$ variables
	yields a satisfying assignment for the boolean formula~$\phi$.
	Let us fix one of these
	satisfying assignments to $x_1,\ldots,x_n$.
	From this fixed assignment, we get the
	vector 
	$(d_1,\ldots,d_n)$
	as defined in the proof of
	Lemma~\ref{lem:correspondence-assignments-weights}, where the first
	$k$ positions correspond to the players $M'\subseteq M$, $\|M'\|=k$,
	which we add to the game~$\mathcal{G}$.
	
	Since there are more than $2^{n-k-1}$ assignments to
	$x_{n-k},\ldots,x_n$ which---together with the fixed assignments to
	$x_1,\ldots,x_k$---satisfy~$\phi$, by
	Lemma~\ref{lem:correspondence-assignments-weights} there are more
	than $2^{n-k-1}$ subsets of $A\cup C \cup M'$ such that the players' weights
	in each subset sum up to~$q_1$. 
	Each of these subsets with total weight $q_1$   
	can form 
	coalitions of weight $q-1$
	with each player from $W$ having weight
	$q-q_1-(\ell+1)$, $\ell\in\{1,\ldots,k\}$, and $\ell$ weight-$1$
	players from $X$---and there are ${k \choose \ell}$ such coalitions.
	Therefore, recalling from Case~2 above that $Y \cup Z$ already
	contains $2^{n-k-1}$ coalitions of weight $q-1$, we have
	\begin{align*}
		\PenroseBanzhaf(\mathcal{G}_{\cup M'},p) &
		> \frac{2^{n-k-1}+2^{n-k-1} \sum_{\ell=1}^{k}{k \choose \ell}}{2^{\|N\|+k-1}} \\
		& = \frac{2^{n-k-1}+(2^{k}-1)\cdot 2^{n-k-1}}{2^{\|N\|+k-1}} \\
		& = \frac{2^{k}\cdot 2^{n-k-1}}{2^{\|N\|+k-1}}
		~=~ \frac{2^{n-k-1}}{2^{\|N\|-1}}
		~=~
		\PenroseBanzhaf(\mathcal{G},p),          
	\end{align*}
	so player~$p$'s Penrose--Banzhaf index is strictly larger in the new
	game~$\mathcal{G}_{\cup M'}$ than in the old game~$\mathcal{G}$,
	i.e., we have constructed a yes-instance of our control problem.

	\proofif
	Assume now that $(\phi, k)$ is a no-instance of \textsc{E-MajSAT},
	i.e., there does not exist any
	assignment to the variables $x_1,\ldots,x_k$ such that a majority
	of assignments to the remaining $n-k$ variables satisfies the boolean
	formula~$\phi$.  In other words, for each assignment to $x_1,\ldots,x_k$,
	there exist at most $2^{n-k-1}$ assignments to $x_{k+1},\ldots,x_n$
	that yield a satisfying assignment for~$\phi$.
	Again, we consider subsets $M'\subseteq M$ of players that uniquely
	correspond to the assignments of $x_1,\ldots,x_k$ according
	to Lemma~\ref{lem:correspondence-assignments-weights}.
	Note that
	any other possible subset will not
	allow to form new coalitions for which player~$p$ could be pivotal in the new game, i.e., $p$'s Penrose--Banzhaf index will 
	not increase 
	unless we add any player with weight
	either $a_i$ or $b_i$ for each $i\in\{1,\ldots,k\}$.

	By
	Lemma~\ref{lem:correspondence-assignments-weights} and our assumption,
	there are at most $2^{n-k-1}$ subsets of $A\cup C \cup M'$ such that
	the players' weights in each subset sum up to~$q_1$.  As in the proof
	of the ``Only if'' direction, for each $\ell\in\{1,\ldots,k\}$, each
	of these subsets of $A\cup C \cup M'$ forms 
    ${k \choose \ell}$
	coalitions of weight $q-1$
	with a player in $W$ having weight $q-q_1-(\ell+1)$ and $\ell$ players
	in~$X$.
	Again recalling from Case~2 above that $Y \cup Z$ 
	already contains
	$2^{n-k-1}$ coalitions of weight $q-1$, we have
	\begin{align*}
		\PenroseBanzhaf(\mathcal{G}_{\cup M'},p)
		& \le \frac{2^{n-k-1}+(2^{k}-1)\cdot 2^{n-k-1}}{2^{\|N\|+k-1}} \\
		& = \frac{2^{k}\cdot 2^{n-k-1}}{2^{\|N\|+k-1}} 
		~=~ \frac{2^{n-k-1}}{2^{\|N\|-1}} 
		~=~
		\PenroseBanzhaf(\mathcal{G},p).
	\end{align*}
	Thus
	player~$p$'s Penrose--Banzhaf index
	cannot increase by adding up to $k$ players from $M$ to the
	game~$\mathcal{G}$, and we have a no-instance of our control
	problem.~\end{proof}

\begin{theorem}
	\label{thm:adding-increase-nondecrease-SSI} 
	\textsc{Control-by-Adding-Players-to-Increase-$\ShapleyShubik$} and
	\textsc{Control-by-Adding-Players-to-Nondecrease-$\ShapleyShubik$}
	are $\np^{\PP}$-complete.
\end{theorem}
\begin{proof}
	We prove $\np^{\PP}$-hardness of both control problems using  one and the same reduction from $\textsc{E-MajSAT}$ (and argue slightly differently for them).
	Let $(\phi, k)$ be a given instance of $\textsc{E-MajSAT}$, where $\phi$ is a boolean formula in CNF with variables $x_1,\dots,x_n$ and $m$ clauses, and let $k<n$.

	\begin{table*}[h!]
		\caption{\label{tab:adding-increase-nondecrease-SSI}
			Groups of players in the proof of
			Theorem~\ref{thm:adding-increase-nondecrease-SSI}, with their categories,
			numbers, and weights
			(note that, e.g., the sum $\sum_{j=0}^{i-1}\beta_{j}v_{j}$ in the first (size) row has value~$0$ for $i=0$)
		}
		\begin{center}
			\renewcommand{\arraystretch}{1.5}
			\begin{tabular}{c|c|c|c}
				\toprule
				\textbf{Category} & \textbf{Group} & \textbf{Number of Players} & \textbf{Weights} \\
				\midrule
				& distinguished player~$p$ & $1$ & $1$ \\
				\midrule
				(ms) & $A$ & $2n-2k$ & $W_A$ \\
				\midrule
				(ms) & $C$ & $m(r+1)$ & $W_C$ \\
				\midrule 
				(ms) & $C'$ & $m(r+1)$ & $W_{C'}$ \\
				\midrule
				(size) & $D$ & $\delta$ & $1$ \\
				\midrule
				(def) & $S$ & $\sum_{i=1}^{u}(y_i + 1)$ & \makecell{$q-q_{2} - \beta_{j_i} v_{j_i} - j_i v_{i} ' -\delta - 1$ \\  for $i \in\{1,\ldots,u\}$ and $j_i \in\{0,\ldots,y_i\}$} \\ 
				\midrule
				(size) & $V_i$ \par for $i\in\{0,\ldots,y_1\}$ & $\beta_{i}$ & $v_{i}=
				1 + \delta + \sum_{j=0}^{i-1}\beta_{j}v_{j}$  \\
				\midrule
				(num) & $V_{i}'$  \par for $i\in\{1,\ldots,u\}$ & $y_i$ & $v_i ' = (\beta_{y_1}+1)v_{y_1} + \sum_{i'=1}^{i-1} y_{i'}v_{i'}'$ \\
				\midrule
				(def) & $T$ & $2n-2k+1$ & \makecell{$q-\alpha_{i}w^{*}_{i} - iw'-\delta-1$ \\ for $i\in\{0,\ldots,2n-2k\} $} \\
				\midrule
				(size) & $W^{*}_i$ \par for $i\in\{0,\ldots,2n-2k\}$ & $\alpha_i$ &
				$w^{*}_i = (y_u+1)v'_u + \sum_{i'=0}^{i-1}\alpha_{i'}w^{*}_{i'}$ \\
				\midrule
				(num) & $W'$ & $2n-2k$ & $w' = (\alpha_{2n-2k}+1) w^{*}_{2n-2k}$ \\
				\midrule
				& $Z$ & remaining players & $q$ \\
				\bottomrule
			\end{tabular}
		\end{center}
	\end{table*}

	Before we construct an instance of our control problems from $(\phi, k)$, we need to choose some numbers and introduce some notation.
	Let
	\[
	P=6n^2 m + 26n^2 + 8k^2 + 8nm + 18n + 4k - 2m - 3 
	\]
	be the number of players in our game
	(note that $P$ is an odd number).  The numbers
	\begin{eqnarray*}
		\delta & = & 3n^2 m  + 13n^2 + 4k^2 + 3nm + 5n + 4k - 2m - 5, \\
		x & = & \delta + nm + 4n - 2k + m + 3 = \frac{P-1}{2},
		\quad
		\text{ and} \\
		k' & = & \Big(1+\frac{x+1}{P-x}\Big)\cdot\ldots\cdot \Big(1+\frac{x+1}{P-x+k-1}\Big)\le 2^{k}
	\end{eqnarray*}
	with $k' \ge 2$,
	will be used in our calculations later in the proof.  Finally, let
	\[
	z = \ceil{2^{n-k+1}(k'-1)}-1<2^{n+1}
	\]
	and choose $y_1,\ldots,y_u$ with $y_1 > \cdots > y_u$ such that
	\[
	z = 2^{y_1}+\cdots+ 2^{y_u}
	\]
	is satisfied.  Note that $y_1 \le n$ and $u \le n$.

	To make the calculations in our proof simpler, we want all coalitions counted for computing the Shapley--Shubik indices to be equally large (to be more specific, we want these coalitions to have size~$x$).
	Therefore, we define the following values.
	For $i\in \{0,1,\ldots,2n-2k\}$, let
	\[\alpha_i = nm + 4n -2k  + m +2 -i,
	\]
	and for $i\in\{0,\ldots,y_1\}$, let
	\[
	\beta_{i} = (n-r)m + 3n -2k +2 - i.
	\] 
	Finally, let $t' \in \mathbb{N}$ be such that
	\begin{equation*}
		10^{t'} > \max\Big\{2^{\ceil{\log_{2} n}+1}, (2n-2k+1)w' \Big\}
	\end{equation*}
	for $w' = (\alpha_{2n-2k}+1) w^{*}_{2n-2k}$ as defined in Table~\ref{tab:adding-increase-nondecrease-SSI}. 
	For $\phi$, $k$, and $t'$, let $t$, $q_{2}$, $M$, $A$, $C$, and $C'$ with weight lists $W_M$, $W_A$, $W_C$, and $W_{C'}$ be defined as in
	Set~2
	of Definition~\ref{def:prereduction}.
	
	Now, we are ready to construct the instance of our two control problems by adding players to increase or to nondecrease a given player's Shapley--Shubik power index as follows:
	Let $k$ be the limit for the number of players that can be added,
	let $M$ be the set of $2k$ players that can be added and let $W_M$ be the list of their weights,
	let 
	\begin{align*}
		q  = 2 \cdot & \left(w_A + w_M + w_C + w_{C'} 
		+ 10^{t'} + 1\right) 
	\end{align*}
	be the quota of~$\mathcal{G}$, and
	let $N$ be the set of $P$ players in game~$\mathcal{G}$, subdivided into groups
	as presented in Table~\ref{tab:adding-increase-nondecrease-SSI}.

	Note that each group of players in Table~\ref{tab:adding-increase-nondecrease-SSI} (except the distinguished player~$p$
	and group $Z$ whose players are not part of any coalition for which $p$ is pivotal) 
	belongs to some category:
	We categorize players by their function, i.e., there are groups of players who are responsible for defining coalitions that are counted when computing the Shapley--Shubik indices; other groups of players are responsible for the size of the coalition they are in (again, when counted in these indices); and there are players who are responsible for the number of coalitions.  
	Some of these players are defined by setting their weights to the quota minus some values that have to be satisfied by other players (for a sufficiently large quota, so as to make it impossible for the distinguished player to be pivotal for any coalition containing more than one of these players). 
	For the remaining players, we define their weights in such a way that they are not interchangeable.
	
	In more detail, the players with category (def) ``define'' which other players are needed to create a coalition of weight~$q-1$, among the players with category (ms) and the players in~$M$, we will focus on those coalitions whose total weight is $q_{2}$.
	The main purpose of the players from the groups marked (num) is to specify the number of coalitions for which player~$p$ can be pivotal.
	The players from groups with category (size) are used to make all these coalitions of equal size
	(among these players, the players with the same weight are together part of the same coalitions).
	Now, we will discuss the coalitions counted in our proof in detail.
	
	Let us analyze for what coalitions player~$p$ can be pivotal in $\mathcal{G}$ or any new game resulting from $\mathcal{G}$ by adding players from~$M$.
	Player~$p$ is pivotal for coalitions of weight~$q-1$.
	First, note that any two players from $S\cup T$ together have a total weight larger than~$q$.
	Next, the total weight of $N\setminus(\{p\}\cup S \cup T \cup Z)$ 
	is smaller than $q-1$.
	Therefore, a coalition with a total weight of $q-1$ has to contain exactly one of the players in $S\cup T$ and whether this player is in $S$ or $T$ has consequences as to which other players have to be in such a coalition: 
	\begin{description}
		\item[Case 1:] If the coalition contains a player from $S$, it also has to contain the players from
		$M \cup A \cup C \cup C'$ whose weights sum up to $q_{2}$, 
		some players from $V_i \cup V_{i}'$
		(for $i$ defined as in Table~\ref{tab:adding-increase-nondecrease-SSI}),
		and all players from~$D$---the players from
                \[
                \bigcup_{i=0}^{y_1}V_i \cup \bigcup_{i=1}^{u}V_i '\cup \bigcup_{i=0}^{2n-2k} W^{*}_i \cup W' \cup D   
		\]
                have total weight smaller than $10^{t'}$.
		Therefore, $q_{2}$ can be achieved only by the players from $M \cup A \cup C \cup C'$.  Recalling that $q_{2}$ can be achieved by a set consisting of those players whose weights take exactly one value in $\{a_i,b_i\}$ for each $i\in\{1,\ldots,n\}$,
		we have to add a set $M' \subseteq M$ with $\|M'\|=k$ to~$\mathcal{G}$.
		But weights of players from $M \cup A \cup C \cup C'$ can sum up only to values which are divisible by $10^{t'}$ therefore they can achieve only the $q_{2}$-part. Each player from $\bigcup_{i=0}^{2n-2k} W^{*}_i \cup W'$ also
		is too heavy to achieve the
		required value. 
		
		\item[Case 2:] If the coalition contains a player from $T$, the coalition also has to contain
		some of the players from $W^{*}_i \cup W'$
		and all players from $D$.
		Also here, we do not find any other combination of players which could form a weight-$(q-1)$ coalition with a player in~$T$---all players in
                \[
                \bigcup_{i=0}^{y_1}V_i \cup \bigcup_{i=1}^{u}V_i ' \cup D   
		\]
                have a total weight too small to be able to replace even one player from
		$\bigcup_{i=0}^{2n-2k} W^{*}_i\cup W'$ and (as mentioned in Case~1) any player in $M\cup A \cup C \cup C'$ together with any player from $T$ has total weight larger than $q-1$. 
	\end{description}

	In both cases, each coalition has the same size of 
	\[
	1+\delta + n+m(r+1)+\beta_{j} + j = 1+\delta +\alpha_i + i = x  
	\]
	for any $i\in \{0,\ldots,2n-2k\}$
	and $j\in\{0,\ldots,y_1\}$.
	
	Since there are no players with weights $a_i$ or $b_i$ for $i\in\{1,\ldots,k\}$ in game~$\mathcal{G}$, player~$p$ can be pivotal only for the coalitions described in the second case above and therefore,
	\[
	\ShapleyShubik(\mathcal{G},p)=2^{2n-2k}\frac{x!(P-x-1)!}{P!}.
	\]

	To prove the correctness of the reduction, we show that the following
	three statements are pairwise equivalent:
	\begin{itemize}
		\item $(\phi, k)$ is a yes-instance of $\textsc{E-MajSAT}$;
		\item $(\mathcal{G}, M, p, k)$ is a yes-instance of
		\textsc{Control-by-Adding-Players-to-Increase-}$\ShapleyShubik$;
		\item $(\mathcal{G}, M, p, k)$ is a yes-instance of
		\textsc{Control-by-Adding-Players-to-Non\-de\-crease-}$\ShapleyShubik$.
	\end{itemize}

	Suppose $(\phi, k)$ is a yes-instance of $\textsc{E-MajSAT}$,
	i.e., there exists an assignment to $x_1,\ldots,x_k$ such that a (strict)
	majority of assignments of the remaining $n-k$ variables satisfies
	the boolean formula~$\phi$.
	Let us fix one of these satisfying assignments.
	From this fixed assignment, we get the
	vector $\vec{d} = (d_1,\ldots,d_n)$ as defined in the proof of
	Lemma~\ref{lem:correspondence-assignments-weights}, where the first
	$k$ positions correspond to the players in $M'\subseteq M$, $\|M'\|=k$,
	which we add to the game~$\mathcal{G}$.
	
	Since there are at least $2^{n-k-1}+1$ assignments for
	$x_{n-k},\ldots,x_n$ which---together with the fixed assignments for
	$x_1,\ldots,x_k$---satisfy~$\phi$, by
	Lemma~\ref{lem:correspondence-assignments-weights} there are more than $2^{n-k-1}$ subsets of $M' \cup A \cup C \cup  C'$ such that the players' weights
	in each subset sum up to~$q_{2}$. 
	Now, each of these subsets can form $2^{y_1}+\cdots+2^{y_u}=z$ coalitions with the players from
        \[
        S \cup \bigcup_{i=0}^{y_1}V_i \cup \bigcup_{i=1}^{u}V_i ' \cup D
	\]
        for which player~$p$ is pivotal in the new game~$\mathcal{G}_{\cup M'}$.
	Therefore, 
	\begin{align*}
		\ShapleyShubik(& \mathcal{G}_{\cup M'},p) \\ \ge & ~ \Big(2^{2n-2k}+z\cdot (2^{n-k-1}+1)\Big)\frac{x!(P+k-1-x)!}{(P+k)!} \\
		= & ~ \Big(2^{2n-2k}+\left(\ceil{2^{n-k+1}(k'-1)}-1\right)\cdot \left(2^{n-k-1}+1\right)\Big) \\
		& ~ \cdot\frac{x!(P-1-x)!}{P!} \cdot \frac{(P-x) \cdots (P+k-1-x)}{(P+1) \cdots (P+k)} \\ 
		\ge & ~ \Big(2^{2n-2k}+\left(2^{n-k+1}(k'-1)-1\right)\cdot \left(2^{n-k-1}+1\right)\Big) \\
		& ~ \cdot\frac{1}{k'}\frac{x!(P-1-x)!}{P!} \\
		= & ~ \Big(2^{2n-2k}k'- 2^{n-k-1}+2^{n-k+1}(k'-1)-1\Big) \\
		& ~ \cdot\frac{1}{k'}\cdot \frac{x!(P-1-x)!}{P!} \\
		> & ~ \ShapleyShubik(\mathcal{G},p),
	\end{align*}
	so player~$p$'s Shapley--Shubik power index is strictly larger in the new
	game~$\mathcal{G}_{\cup M'}$ than in the old game~$\mathcal{G}$,
	i.e., we have constructed a yes-instance of both our control problems.

	Conversely, suppose now that $(\phi, k)$ is a no-instance of
	$\textsc{E-MajSAT}$,
	i.e., for each assignment to $x_1,\ldots,x_k$,
	there exist at most $2^{n-k-1}$ assignments of $x_{k+1},\ldots,x_n$
	which satisfy~$\phi$.
	It is enough to consider subsets $M'\subseteq M$ of players that uniquely
	correspond to the assignments of $x_1,\ldots,x_k$ according
	to Lemma~\ref{lem:correspondence-assignments-weights}, because any other possible subset will not allow to form new coalitions for which player~$p$ could be pivotal in the new game, i.e., $p$'s Shapley--Shubik index will only decrease if we do not add any player with weight either $a_i$ or $b_i$ for each $i\in\{1,\ldots,k\}$.
	
	Now let $M'\subseteq M$ be any subset of players that corresponds to
	some assignment to $x_1,\ldots,x_k$.  By
	Lemma~\ref{lem:correspondence-assignments-weights} and our assumption,
	there are at most $2^{n-k-1}$ subsets of $M' \cup A \cup C \cup C'$ such that
	the players' weights in each subset sum up to~$q_{2}$.
	For each of these sets, there are exactly $z$ new coalitions 
	described in Case~$1$
	for which $p$ is pivotal after adding the new players from~$M'$.
	Therefore, 
	\begin{align*}
		\ShapleyShubik( & \mathcal{G}_{\cup M'},p) \\ 
		\le 
		& ~ \Big(2^{2n-2k}+\left(\ceil{2^{n-k+1}(k'-1)}-1 \right)\cdot 2^{n-k-1}\Big) \\ 
		& ~ \cdot\frac{x!(P-1-x)!}{P!}\cdot \frac{(P-x)\cdots (P+k-1-x)}{(P+1)\cdots (P+k)}\\
		< & ~ \Big(2^{2n-2k}+2^{n-k+1}(k'-1)\cdot 2^{n-k-1}\Big) \\ 
		& ~ \cdot\frac{1}{k'}\cdot \frac{x!(P-1-x)!}{P!} \\
		= & ~ \frac{2^{2n-2k}k'}{k'}\cdot \frac{x!(P-1-x)!}{P!} %
		~=~  \ShapleyShubik(\mathcal{G},p),
	\end{align*}
	which means that the Shapley--Shubik index of player~$p$ decreases.
	Thus the Shapley--Shubik index of player~$p$ can neither increase nor nondecrease by adding up to $k$ players from $M$ to the game~$\mathcal{G}$, and we have a no-instance of both our control problems.
\end{proof}

\begin{theorem}\label{restofresults}
	The following problems are $\np^{\PP}$-complete:
\begin{description}
  \item[(a)]\label{nondecrease-PBI}
\textsc{Control-by-Adding-Players-to-Nondecrease-$\PenroseBanzhaf$}. \\
And for $\gamma\in\{\PenroseBanzhaf,\ShapleyShubik\}$,
  \item[(b)]\label{decrease-PBI} 
	\textsc{Control-by-Adding-Players-to-Decrease-$\gamma$},
  \item[(c)]\label{nonincrease-PBI} 
	\textsc{Control-by-Adding-Players-to-Nonincrease-$\gamma$}, and
  \item[(d)]\label{maintain-PBI} 
	\textsc{Control-by-Adding-Players-to-Maintain-$\gamma$}.
\end{description}
\end{theorem}

\section{Conclusions}

We have shown that control by adding players to WVGs so as to
change or maintain a given player's Shapley--Shubik or
Penrose--Banzhaf index is $\np^{\PP}$-complete, thus settling the
complexity of these problems by raising their lower bounds so as to
match their upper bound.  Compared with the eminently rich body of
results on control attacks in
voting~\cite{fal-rot:b:handbook-comsoc-control-and-bribery}, these
results fill a glaring gap in the literature on 
WVGs
which---perhaps due to the immense hardness of these problems that is
proven here---fairly much has neglected issues of control attacks to
date.

For future work, we propose to study the corresponding problems for
deleting players from WVGs.
Further, it would be interesting to study these problems in the model
proposed by Kaczmarek and
Rothe~\citep{kac-rot:j:controlling-weighted-voting-games-by-deleting-or-adding-players-with-or-without-changing-the-quota}
in which the quota is indirectly changed when players are added or deleted.
Our techniques may also turn out to be
useful for closing the complexity gap for other problems in
$\np^{\PP}$ only known to be $\PP$-hard, such as false-name
manipulation~\cite{azi-bac-elk-pat:j:false-name-manipulations-wvg,rey-rot:j:false-name-manipulation-PP-hard}
and control by adding or deleting edges in graph-restricted
WVGs~\cite{kac-rot-tal:c:complexity-of-control-by-adding-or-deleting-edges-in-graph-restricted-weighted-voting-games}.

\begin{ack}
  We thank the reviewers for their helpful comments.
  This work was supported in part by Deutsche Forschungsgemeinschaft under grants RO~1202/21-1 and RO~1202/21-2 (project number 438204498).
\end{ack}

\bibliography{joannabibliography}

\clearpage

\appendix

\section*{Technical Appendix}

\medskip
\sproofof{Lemma~\ref{lem:e-exact-sat}}
The \emph{counting (polynomial-time) hierarchy} was introduced by
Wagner~\citep{wag:j:succinct} and, independently, by Parberry and
Schnitger~\citep{par-sch:j:parallel-computation-with-threshold-functions}.
Wagner~\citep{wag:j:succinct} characterized the levels of this
hierarchy via the \emph{counting quantifier $\CC$} and the \emph{exact
	counting quantifier~$\CE$}.  Applying the former to a $\p$ predicate
yields the complexity class $\PP$, and applying the latter to a $\p$
predicate yields the complexity class~$\CE \p$:
$A\in \CE \p$ if and only if there exists a set $B \in \p$, a
polynomial-time computable function~$f$, and a polynomial $p$ such
that
\[
x \in A \iff \|\{y \condition |y|\le p(|x|) \wedge
(x,y)\in B\}\|=f(x).
\]

$\textsc{Exact-SAT} = \{(\phi,n) \condition \phi$ is a boolean	formula with exactly $n$ satisfying truth assignments$\}$ is a
typical $\CE \p$-complete problem.

Continuing the research on the counting hierarchy,
Tor\'{a}n~\citep{tor:j:quantifiers} showed that
for any class $\mathcal{K}$ in the counting hierarchy,
\begin{enumerate}
	\item $\exists \CC \mathcal{K} = \exists \CE \mathcal{K}$,
	\item $\PP^{\mathcal{K}} = \CC \mathcal{K}$, and
	\item $\np^{\CC \mathcal{K}} = \exists \CC \mathcal{K}$.
\end{enumerate}
In particular, for $\mathcal{K} = \p$, we have $\exists \CC \p =
\exists \CE \p$ and $\np^{\PP} = \exists \CE \p$. 

The proof 
of the $\np^\PP$-completeness of $\textsc{E-ExaSAT}$
is analogous to the
proof of $\np^{\PP}$-completeness of \textsc{E-MajSAT} (and of
\textsc{E-MinSAT}), but instead of using the $\PP$-complete problem
\textsc{MajSat} (or its complement), we use the $\CE \p$-complete
problem \textsc{Exact-SAT}.~\eproofof{Lemma~\ref{lem:e-exact-sat}}

In some of our proofs, we will use slightly different sets of weight
vectors than Sets~1 and~2 from Definition~\ref{def:prereduction}, so
we now slightly modify Definition~\ref{def:prereduction} by defining
Sets~3 and~4 instead of Sets~1 and~2.

\begin{definition}
	\label{def:prereduction2}
	Let $\phi$ be given boolean formula in CNF
	with variables $x_1,\dots,x_n$ and $m$ clauses.
	Let $k\in\mathbb{N}$ with $k\le n$ and $r = \ceil{\log_{2} n}-1$.
	Let us define the following two sets of weight vectors which are going to be assigned as weights to players divided either into three sets---$M$, $A$, and $C$---or into four sets---$M$, $A$, $C$, and $C'$---in our proofs later on:
	\begin{description}
		\item[Set 3:] For some $t \in\mathbb{N}\setminus\{0\}$ such that
		$10^{t} > 2^{\ceil{\log_2 n}+1}$, 
		and
		for $i \in \{1,\ldots,n\}$, define
		\begin{eqnarray*}
			a_i & = & 1 + 10^{t(m+1)+2i}+\sum_{\substack{j\,:\ \textrm{clause $j$} \\ \textrm{ contains $x_i$}}} 10^{tj}
			\text{ and } \\
			b_i & = & 1 + 10^{t(m+1)+2i}+\sum_{\substack{j\,:\ \textrm{clause $j$} \\ \textrm{ contains $\neg x_i$}}} 10^{tj},
		\end{eqnarray*}
		and for $j\in\{1,\ldots,m\}$ and $s\in\{0,\ldots,r\}$, define
		\begin{eqnarray*}
			c_{j,s} & = & 2^{s}\cdot 10^{tj}.
		\end{eqnarray*}
		Define the following three weight vectors:
		\begin{eqnarray*}
			W_M & = & (a_1,\ldots,a_k,b_1,\ldots,b_k), \\
			W_A & = & (a_{k+1},\ldots,a_n,b_{k+1},\ldots,b_n),\\
			W_C & = & (c_{1,0},\ldots,c_{m,r}).
		\end{eqnarray*}
		
		\item[Set 4:] For some $t,t' \in\mathbb{N}\setminus\{0\}$ such that
		$10^{t'} > 2^{\ceil{\log_2 n}+1}$ 
		and 
		$10^t> 10^{t'} + 2^{\ceil{\log_2 n}+1}\sum_{l=1}^{m}10^{lt'}$, 
		and
		for $i \in \{1,\ldots,n\}$, define $a_i$ and $b_i$ as in Set~3,
		\OMIT{
			\begin{eqnarray*}
				a_i & = & 1 + 10^{t(m+1)+2i}+\sum_{\substack{j\,:\ \textrm{clause $j$} \\ \textrm{ contains $x_i$}}} 10^{tj}
				\text{ and } \\
				b_i & = & 1 + 10^{t(m+1)+2i}+\sum_{\substack{j\,:\ \textrm{clause $j$} \\ \textrm{ contains $\neg x_i$}}} 10^{tj},
			\end{eqnarray*}
		} %
		and for $j\in\{1,\ldots,m\}$ and $s\in\{0,\ldots,r\}$, let
		\[
		c_{j,s}'  =  2^{s}\cdot 10^{t'j}
		\text{ and } 
		c_{j,s}  =  2^{s}\cdot 10^{tj}+c_{j,s}'.
		\]
		In addition to $W_M$ and $W_A$ defined as in Set~3,
		define the following two weight vectors: 
		\[
		W_{C'}  =  (c_{1,0}',\ldots,c_{m,r}')
		\text{ and }
		W_C  =  (c_{1,0},\ldots,c_{m,r}).
		\]
	\end{description}	
	Additionally, let
	\begin{align*}
		q_{3}  ~=~ & n + \sum_{i=1}^{n} 10^{t(m+1)+2i} + 2^{\ceil{\log_{2} n}}\sum_{j=1}^{m} 10^{tj} \textrm{ and }\\
		q_{4}  ~=~ & n + \sum_{i=1}^{n} 10^{t(m+1)+2i} + 2^{\ceil{\log_{2} n}}\sum_{j=1}^{m} 10^{tj}\\
		& + \left(2^{\ceil{\log_{2} n}}-1 \right)\sum_{j=1}^{m} 10^{t'j}.
	\end{align*}	
\end{definition}

We now provide a detailed proof of
Lemma~\ref{lem:correspondence-assignments-weights} (which was only
sketched in the main paper).  Note that the analogue of
Lemma~\ref{lem:correspondence-assignments-weights} with Sets~3 and~4
from Definition~\ref{def:prereduction2} replacing Sets~1 and~2 from
Definition~\ref{def:prereduction} can be shown analogously to the
proof of Lemma~\ref{lem:correspondence-assignments-weights}.  So,
while proving Lemma~\ref{lem:correspondence-assignments-weights}, we
also provide these analogous proof details in parallel.
\medskip

\sproofof{Lemma~\ref{lem:correspondence-assignments-weights}}
Let us start with analyzing which subsets of $M\cup A\cup C$
for $i\in\{1,3\}$ or $M \cup A \cup C \cup C'$ for
$i\in\{2,4\}$ can achieve a total weight of~$q_i$.
The summand of~$q_i$ from the interval
\begin{itemize}
	\item for $i\in\{1,2\}$:
	\[
	\Big[\sum_{j=1}^{n} 10^{t(m+1)+j}, q_i \Big];
	\]
	\item for $i\in\{3,4\}$:
	\[
	\Big[n + \sum_{j=1}^{n} 10^{t(m+1)+2j}, q_i \Big]
	\]
\end{itemize}
can be achieved only by player sets containing some players from
$M\cup A$, because all players with a smaller weight together are
not heavy enough: For $i\in\{1,3\}$, we have
\begin{align*}
	& w_C + 10^t \\
	& = (2^{\ceil{\log_2 n}}-1)\sum_{j=1}^m 10^{t j}  + 10^t \\
	& < (2^{\ceil{\log_2 n}}-1)10^{t} + (2^{\ceil{\log_2 n}}-1)\sum_{j=2}^m 10^{t j} + 10^t \\
	& = 2^{\ceil{\log_2 n}} \cdot 10^{t} + (2^{\ceil{\log_2 n}}-1) 10^{2t} \\
	& ~~~~ +  (2^{\ceil{\log_2 n}}-1)\sum_{j=3}^m 10^{t j} \\
	& < 2^{\ceil{\log_2 n}} \cdot 10^{2t} +  (2^{\ceil{\log_2 n}}-1)\sum_{j=3}^m 10^{t j}
\end{align*}
\begin{align*}
	& < 2^{\ceil{\log_2 n}} \cdot  10^{t(m-1)} +  (2^{\ceil{\log_2 n}}-1) 10^{tm} \\
	& < 2^{\ceil{\log_2 n}} \cdot 10^{tm}  \\
	& <  10^{t(m+1)+1},
\end{align*}
and for $i\in\{2,4\}$, we have
\begin{align*}
	& w_C + w_{C'} + 10^{t'} \\ 
	& < w_C + 10^t \\
	& = (2^{\ceil{\log_2 n}}-1)\sum_{j=1}^m 10^{t j}  + w_{C'} + 10^t \\
	& < (2^{\ceil{\log_2 n}}-1)\sum_{j=1}^m 10^{t j} + 2 \cdot 10^t \\
	& = (2^{\ceil{\log_2 n}} + 1) \cdot 10^{t} + (2^{\ceil{\log_2 n}}-1) 10^{2t} \\
	& ~~~~ +  (2^{\ceil{\log_2 n}}-1)\sum_{j=3}^m 10^{t j} \\
	& < 2^{\ceil{\log_2 n}} \cdot 10^{2t} +  (2^{\ceil{\log_2 n}}-1)\sum_{j=3}^m 10^{t j} \\
	& < 2^{\ceil{\log_2 n}} \cdot 10^{tm}  \\
	& <  10^{t(m+1)+1}.
\end{align*}

Moreover, $q_i$ can be achieved only by subsets containing \emph{exactly}
$n$ players from $M\cup A$, exactly one weight from each pair $\{a_j, b_j\}$,
$j\in\{1,\ldots,n\}$, because for each~$j\ge 2$, if $i\in\{1,2\}$, we have
\begin{align*}
	& 10^{t(m+1)+1} + \sum_{l=1}^{j-1} (a_l + b_l) \\
	& < 4\cdot 10^{t(m+1)+1} + 3 \cdot \sum_{l=2}^{j-1} 10^{t(m+1)+l} \\
	& < 10^{t(m+1)+2} + 3 \cdot 10^{t(m+1)+2}  + 3 \cdot \sum_{l=3}^{j-1} 10^{t(m+1)+l} \\
	& < 4 \cdot 10^{t(m+1)+j-2} + 3 \cdot 10^{t(m+1)+j-1} \\
	& < 4 \cdot 10^{t(m+1)+j-1} \\
	& < 10^{t(m+1)+j},
\end{align*}
and if $i\in\{3,4\}$, we have
\begin{align*}
	& 10^{t(m+1)+1} + \sum_{l=1}^{j-1} \Big(a_l + b_l + 9\cdot 10^{t(m+1)+2l}\Big) \\
	& < 2\cdot 10^{t(m+1)+1} + 11\cdot 10^{t(m+1)+2}  + 12 \cdot \sum_{l=2}^{j-1} 10^{t(m+1)+2l} \\
	& < 10^{t(m+1)+4} + 12 \cdot 10^{t(m+1)+4}  + 12 \cdot \sum_{l=3}^{j-1} 10^{t(m+1)+2l} \\
	&  < 13 \cdot 10^{t(m+1)+2j-4} + 12 \cdot 10^{t(m+1)+2j-2} \\
	&  < 13 \cdot 10^{t(m+1)+2j-2} \\
	& < 10^{t(m+1)+2j},
\end{align*}
and both $a_l$ and $b_l$ are together too large to
satisfy $10^{t(m+1)+l}$ and $10^{t(m+1)+2l}$, respectively, 
with 
any other
smaller part of this sum.
Therefore, there are exactly $2^{n}$
subsets of $M\cup A$ which---jointly with some players from $C$ or
$C\cup C'$---can achieve the value of~$q_i$, for $i\in\{1,2,3,4\}$.

Now, let us fix one of the subsets of $M\cup A$ mentioned above. Its
weight
\begin{itemize}
	\item for $i\in\{1,2\}$ is:
	\[
	\sum_{j=1}^{n} 10^{t(m+1)+j} + \sum_{j=1}^{m}p_{j} 10^{tj} 
	\]
	\item and for $i\in\{3,4\}$ it is:
	\[
	\sum_{j=1}^{n} 10^{t(m+1)+2j} + n + \sum_{j=1}^{m}p_{j} 10^{tj}
	\]
\end{itemize}
for some $p_j \in \{0,\ldots,n\}$ for each $j\in\{1,\ldots,m\}$. 
\OMIT{
	(the
	fact that $p_j$ cannot be larger than $n$ comes from the assumption
	that there is no clause in $\phi$ such that it contains both a
	variable and its negation, i.e., for each $j\in\{1,\ldots,m\}$ and
	each $l\in\{1,\ldots,n\}$, the value $10^{tj}$ cannot be a part of
	both $a_l$ and $b_l$ at the same time).  
} %
To achieve~$q_i$, we still
need players whose total weight
\begin{itemize}
	\item for $i\in\{1,3\}$ is:
	\[
	\sum_{j=1}^{m}(2^{\ceil{\log_{2} n}}-p_{j}) \cdot 10^{tj} 
	\]
	\item and for $i\in\{2,4\}$ it is:
	\[
	\sum_{j=1}^{m}(2^{\ceil{\log_{2} n}}-p_{j}) \cdot 10^{tj} + \sum_{j=1}^{m}(2^{\ceil{\log_{2} n}}-1) \cdot 10^{t'j}.
	\]
\end{itemize}
In the case of~$i\in\{2,4\}$, if $p_j<2^{\ceil{\log_{2} n}}$, the value from the interval
\begin{align*}
	\Big[ & \sum_{j=1}^{m}(2^{\ceil{\log_{2} n}}-p_{j}) \cdot 10^{tj}, \\ & \sum_{j=1}^{m}(2^{\ceil{\log_{2} n}}-p_{j}) \cdot 10^{tj} + \sum_{j=1}^{m}(2^{\ceil{\log_{2} n}}-1) \cdot 10^{t'j}\Big]
\end{align*}
can be achieved only by those subsets that contain some players
from~$C$, since
\[
w_{C'} + 10^{t'} < 10^{t}.
\] 
So, let us consider the players from~$C$ now. For any $j\in\{1,\ldots,m\}$,
\begin{itemize}
	\item the value $\left(2^{\ceil{\log_{2} n}}-p_{j}\right) \cdot 10^{tj}$ for
	$i\in\{1,3\}$ and
	
	\item the value $\left(2^{\ceil{\log_{2} n}}-p_{j}\right) \cdot
	10^{tj} + \left(2^{\ceil{\log_{2} n}}-p_{j}\right) \cdot 10^{t'j}$
	for $i\in\{2,4\}$ 
\end{itemize}
can be achieved only by players from
$\{c_{j,0},\ldots,c_{j,r}\}$.\footnote{Note that in the latter case,
	another possibly smaller summand with the unchanged larger summand
	cannot be achieved without players from $C'$, i.e., for any subset
	of~$C$, its weight is also in the form $\sum_{l=1}^{m}
	h_l \cdot \left(10^{tl} + 10^{t'l}\right)$.}
This is true because any player with weight
\begin{align*}
	c_{j+1,l} & \ge 10^{t(j+1)} \\
	& = 10^{t}\cdot 10^{tj} \\ 
	& > 2^{\ceil{\log_2 n}+1}\cdot 10^{tj} \\ 
	& = (2^{\ceil{\log_2 n}}-1) 10^{tj} + (2^{\ceil{\log_2 n}}+1) 10^{tj} \\
	& \ge (2^{\ceil{\log_2 n}}-1) 10^{tj} + (2^{\ceil{\log_2 n}}-1) 10^{tj} + 2\cdot 10^{t}
\end{align*}    
has greater weight than all players from $\{c_{j,0},\ldots,c_{j,r}\}$
with all players of smaller weight together.  Analogously, all players
of weight smaller than $10^{tj}$ together have total weight smaller
than this value, i.e., for each~$j$,
\begin{equation}\label{addingplayers-increasePBI:sum:c_js}
	\sum_{s=0}^{r} 2^{s}\cdot 10^{tj} = \left(2^{\ceil{\log_{2} n}}-1 \right)\cdot 10^{tj}<2^{\ceil{\log_{2} n}}\cdot 10^{tj},
\end{equation}
so for any~$z < 10^t$, we have for $j\ge 2$
(the case of~$j=1$ is straightforward from the definition of~$t$) that
\begin{align*}
	z + \sum_{l=1}^{j-1}\sum_{s=0}^{r} c_{l,s} & < z + \sum_{l=1}^{j-1}2^{\ceil{\log_{2} n}}\cdot 10^{tl} \\
	& < \sum_{l=1}^{j-1}\left(2^{\ceil{\log_{2} n}}+1\right)\cdot 10^{tl} \\
	& \le 10^{tj} 
\end{align*} 
if $i\in\{1,3\}$;
and
if $i\in\{2,4\}$, we have
\begin{align*}
	z + \sum_{l=1}^{j-1}\sum_{s=0}^{r} c_{l,s} & < z + \sum_{l=1}^{j-1}2^{\ceil{\log_{2} n}} \left(10^{tl}+10^{t'l}\right) \\
	& < 2 \cdot 10^t + 2^{\ceil{\log_{2} n}} \sum_{l=1}^{j-1} 10^{tl} \\
	& \le 10^{tj}. 
\end{align*}
Moreover, each subset of $\{c_{j,0},\ldots,c_{j,r}\}$ adds up to a value
that is divisible by~$10^{tj}$, so they cannot achieve any 
value that is not divisible by~$10^{tj}$.
Finally, note that for $p_j = 0$, there exists no
subset with this weight, and for all other possible~$p_j$, there
exists a unique subset of $C$ achieving the value (divided by
$10^{tj}$, the weights or their larger summands correspond to the
binary representation of the number $2^{\ceil{\log_{2} n}}-p_j$). So,
for the fixed subset of $M\cup A$, there exists at most one subset of
$C$ that can be part of the set with weight~$q_i$.

For $i\in\{1,3\}$, we obtain subsets with weight~$q_i$.  For
$i\in\{2,4\}$, the subset (if there exists any) of $M\cup A \cup C$
for the fixed subset of $M\cup A$ has weight
\[
\sum_{j=1}^{n} 10^{t(m+1)+j} + 2^{\ceil{\log_{2} n}} \sum_{j=1}^{m} 10^{tj} + \sum_{j=1}^{m}\left(2^{\ceil{\log_{2} n}}-p_{j}\right) \cdot 10^{t'j}
\] 
for $i=2$, and it has weight
\[
\sum_{j=1}^{n} 10^{t(m+1)+2j} + 2^{\ceil{\log_{2} n}} \sum_{j=1}^{m} 10^{tj} + \sum_{j=1}^{m}\left(2^{\ceil{\log_{2} n}}-p_{j}\right) \cdot 10^{t'j}
\] 
for $i=4$, i.e., we need some players from $C'$ with a total weight of
$\sum_{j=1}^{m}(p_{j}-1) \cdot 10^{t'j}$. Analogously to the case of
players from~$C$ (but always unlike the previous case), there exists a
unique subset of $C'$ with that weight.

To sum up, each set $S$ of weight~$q_i$, for $i\in\{1,2,3,4\}$, has to
contain exactly $n$ players from $M\cup A$ (namely, $n$ players, each
with exactly one weight from $\{a_j,b_j\}$, $j\in\{1,\ldots,n\}$), and
for each $S\cap (M\cup A)$, there exists exactly one set of
weight~$q_i$ (but there can exist subsets of $M\cup A$ of the
mentioned form that are not contained in any set of weight~$q_i$).

Let us now prove that there exists a bijection between the sets of
weight~$q_i$ and the set of value assignments to the variables
$x_1,\ldots,x_n$ satisfying the given formula~$\phi$.

For each value
assignment to the variables $x_1,\ldots,x_n$, let
$1$ represent \texttt{true} and $0$ \texttt{false}, and let
\begin{equation}
	d_l = \left\{
	\begin{array}{cc}
		a_l & \textrm{ if } x_l = 1, \\
		b_l & \textrm{ if } x_l = 0.
	\end{array}
	\right. 
\end{equation}
The resulting weight vector $\vec{d} = (d_1,\ldots,d_n)$ is unique for
each assignment to $x_1,\ldots,x_n$ (from the previously mentioned
assumption that no clause contains both a variable and its negation,
so $a_l \neq b_l$ for any $l\in\{1,\ldots,n\}$).  Also, if this
vector~$\vec{d}$ corresponds to a satisfying assignment of~$\phi$, the
total weight of the players' subset equals
\begin{itemize}
	\item for Set~1 and Set~2:
	\[
	\sum_{l=1}^{n}d_l = \sum_{l=1}^{n} 10^{t(m+1)+l} + \sum_{j=1}^{m}p_{j} 10^{tj},
	\]
	\item and for Set~3 and Set~4:
	\[
	\sum_{l=1}^{n}d_l = n + \sum_{l=1}^{n} 10^{t(m+1)+2l} + \sum_{j=1}^{m}p_{j} 10^{tj},
	\]
\end{itemize}
where $p_j$, $1\le p_j \le n$, is at least $1$ since each clause is
satisfied by our fixed assignment: For each clause~$j$, there exists
some $x_l$ making it true (i.e., either $x_l=1$ and the clause~$j$
contains~$x_l$, or $x_l=0$ and~$j$ contains~$\neg x_l$), which implies
that the corresponding $d_l$ has $10^{tj}$ as one of its summands
(i.e., either $d_l=a_l$ if $x_l$ is contained in clause~$j$, or
$d_l=b_l$ if $\neg x_l$ is contained in~$j$).  Because $p_j\neq 0$ for
all $j\in\{1,\ldots,m\}$, from the previous analysis, there exists
exactly one subset of 
$C$ when $i\in\{1,3\}$ or exactly one subset of
$C\cup C'$ when $i\in\{2,4\}$ such that the players with the
corresponding weights together with the players whose weights
correspond to $\vec{d}$ form a coalition of weight~$q_i$. Therefore,
for each value assignment satisfying~$\phi$, there exists a unique set
of players from
$A\cup M \cup C$ (respectively, $A\cup M \cup C \cup C'$)
with total weight~$q_i$.

Conversely, let $S\subseteq M \cup A \cup C$ for $i\in\{1,3\}$, and
$S\subseteq M \cup A \cup C \cup C'$ for $i\in\{2,4\}$, be a coalition
of players whose total weight is~$q_i$.  From the previous analysis,
$S$ can contain exactly one player with weight from $\{a_j, b_j\}$ for
$j\in\{1,\ldots,n\}$, and for $S\cap (M\cup A)$, there exists exactly
one subset of $C$ for $i\in\{1,3\}$, and exactly one subset of $C\cup
C'$ for $i\in\{2,4\}$, which creates with the former a coalition of
players with total weight~$q_i$, i.e., there exist no two different
sets $S$ and $S'$ both with $w_S = w_{S'} = q_i$ such that
$S\cap(M\cup A) = S' \cap (M\cup A)$.

For the set $S\cap (M\cup A)$ with the weight vector
$(d_1,\ldots,d_n)$, set
\begin{equation}\label{addingplayers-increasePBI:def:d_i}
	x_\ell = \left\{
	\begin{array}{ll}
		1 & \textrm{ if } d_\ell = a_\ell \\
		0 & \textrm{ if } d_\ell = b_\ell
	\end{array}
	\right. 
\end{equation}
for $\ell\in\{1,\ldots,n\}$.
For each clause $j\in \{1,\ldots,m\}$,
there exists some $d_\ell$ corresponding to the 
player whose weight's 
part is equal to $10^{tj}$; and if the
weight
is~$a_\ell$, clause $j$ contains~$x_\ell$, so assigning
\texttt{true} to $x_\ell$ makes clause $j$ true; otherwise, the 
player's weight
is $b_\ell$ and the clause $j$ contains~$\neg x_\ell$, so assigning
\texttt{false} to $x_\ell$ makes $j$ true.  Hence, this is a unique value
assignment to the variables $x_1,\ldots,x_n$ that satisfies~$\phi$ and
is obtained by the described transformation from the set~$S$.~\eproofof{Lemma~\ref{lem:correspondence-assignments-weights}}

We now prove the four statements of Theorem~\ref{restofresults}.
\medskip

\sproofof{Theorem~\ref{restofresults}(a)}
We modify the reduction from the proof of
Theorem~\ref{adding-PBI-increase}. The only change we make is that
the game $\mathcal{G}'$ in our current reduction has one player more
than the game $\mathcal{G}$ does, i.e., $\mathcal{G}'$ has two
players (instead of one) with weight $q-1$
in group~$Y$.
Let $N'$ with $\|N'\|=\|N\|+1$
be the corresponding player set of~$\mathcal{G}'$.

Therefore,
in the new game $\mathcal{G}'$, we now have
\begin{align*}
	\PenroseBanzhaf(\mathcal{G}',p) & =\frac{2+\sum_{i=1}^{n-k-1}{n-k-1 \choose i}}{2^{\|N'\|-1}} \\ 
	& =\frac{1+\sum_{i=0}^{n-k-1}{n-k-1 \choose i}}{2^{\|N'\|-1}}\\
	&=\frac{2^{n-k-1}+1}{2^{\|N'\|-1}}.
\end{align*}

We now prove the correctness of the reduction: $(\phi, k)$ is a
yes-instance of $\textsc{E-MajSAT}$ if and only if
$(\mathcal{G}', M, p, k)$ is a yes-instance of
\textsc{Control-by-Adding-Players-to-Nondecrease-$\PenroseBanzhaf$}.

\proofonlyif
Suppose that $(\phi, k)$ is a yes-instance of $\textsc{E-MajSAT}$,
i.e., there exists an assignment to $x_1,\ldots,x_k$ such that a
majority of assignments to the remaining $n-k$ variables
yields a satisfying assignment for the boolean formula~$\phi$.
Let us fix one of these
satisfying assignments to $x_1,\ldots,x_n$.
From this fixed assignment, we get the
vector $\vec{d} = (d_1,\ldots,d_n)$ as defined in the proof of
Lemma~\ref{lem:correspondence-assignments-weights}, where the first
$k$ positions correspond to the players $M'\subseteq M$, $\|M'\|=k$,
which we add to the game~$\mathcal{G}'$.

Since there are more than $2^{n-k-1}$ assignments to
$x_{n-k},\ldots,x_n$ which---together with the fixed assignments to
$x_1,\ldots,x_k$---satisfy~$\phi$, by
Lemma~\ref{lem:correspondence-assignments-weights} there are at least $2^{n-k-1}+1$ subsets of $A\cup C \cup M'$ such that the players' weights
in each subset sum up to~$q_1$. 
Each of these subsets with total weight $q_1$   
can form 
coalitions of weight $q-1$
with each player from $W$ having weight
$q-q_1-(\ell+1)$, $\ell\in\{1,\ldots,k\}$, and $\ell$ weight-$1$
players from $X$---and there are ${k \choose \ell}$ such coalitions.
Therefore, we have
\begin{align*}
	\PenroseBanzhaf(\mathcal{G}_{\cup M'}',p) & 
	\ge\frac{2^{n-k-1}+1+(2^{n-k-1}+1) \sum_{\ell=1}^{k}{k \choose \ell}}{2^{\|N'\|+k-1}} \\
	& = \frac{2^{n-k-1}+1+(2^{k}-1) (2^{n-k-1}+1)}{2^{\|N'\|+k-1}} \\
	& = \frac{2^{k} (2^{n-k-1}+1)}{2^{\|N'\|+k-1}} \\
	& = \frac{2^{n-k-1}+1}{2^{\|N'\|-1}}=\PenroseBanzhaf(\mathcal{G}',p),
\end{align*}
so player~$p$'s Penrose--Banzhaf index is not smaller in the new
game~$\mathcal{G}_{\cup M'}'$ than in the old game~$\mathcal{G}'$,
i.e., we have constructed a yes-instance of our control problem.

\proofif
Assume now that $(\phi, k)$ is a no-instance of the problem
$\textsc{E-MajSAT}$, i.e., for each assignment to $x_1,\ldots,x_k$,
there exist
at most $2^{n-k-1}$ assignments to $x_{k+1},\ldots,x_n$
which satisfy~$\phi$.
Analogously to the proof of Theorem~\ref{adding-PBI-increase}, we consider subsets $M'\subseteq M$ of players that uniquely
correspond to the assignments to $x_1,\ldots,x_k$ according
to Lemma~\ref{lem:correspondence-assignments-weights}, because any
other possible subset will not be enough to form new coalitions for
which player~$p$ could be pivotal in the new game, i.e., $p$'s
Penrose--Banzhaf index only decreases in those cases.

Now let $M'\subseteq M$ be any subset of players that corresponds to
some assignment to $x_1,\ldots,x_k$.  By
Lemma~\ref{lem:correspondence-assignments-weights} and our assumption,
there are
at most $2^{n-k-1}$ subsets of $A\cup C \cup M'$ such that
the players' weights in each subset sum up to~$q_1$.  As in the proof
of the ``Only if'' direction, for each $\ell\in\{1,\ldots,k\}$, each
of these subsets of $A\cup C \cup M'$ forms coalitions of weight $q-1$
with a player in $W$ having weight $q-q_1-(\ell+1)$ and $\ell$ players
in $X$---and there are ${k \choose \ell}$ of them.  Therefore, 
\begin{align*}
	\PenroseBanzhaf(\mathcal{G}_{\cup M'}',p)
	& <\frac{2^{n-k-1}+1+(2^{k}-1)\cdot (2^{n-k-1}+1)}{2^{\|N'\|+k-1}}\\
	& =\frac{2^{k}\cdot (2^{n-k-1}+1)}{2^{\|N'\|+k-1}} \\
	& =\frac{2^{n-k-1}+1}{2^{\|N'\|-1}}
	=\PenroseBanzhaf(\mathcal{G}',p),
\end{align*}
which means that the Penrose--Banzhaf index of player~$p$ decreases
also in this case.  Thus the Penrose--Banzhaf index of player~$p$
cannot nondecrease by adding 
at least one and
up to $k$ players from $M$ to the
game~$\mathcal{G}'$, and we have a no-instance of our control
problem.~\eproofof{Theorem~\ref{restofresults}(a)}

\sproofof{Theorem~\ref{restofresults}(b) and (c)}
First, let $\gamma=\PenroseBanzhaf$. 	We will prove $\np^{\PP}$-hardness by providing a reduction from \textsc{E-MinSAT}.
Let $(\phi,k)$ be a given instance of \textsc{E-MinSAT}, where $\phi$ is a boolean formula in CNF with variables $x_1,\dots,x_n$ and $m$ clauses and $2\le k<n$. 

Define
$h'=2k^2$,
$h= (k+1)h'$,
$z=(2n-2k)h$, 
and $e=(2n-k+1)z$, 
and let us choose $t\in\mathbb{N}$ such that 
\begin{equation*}
	10^t > \max\Big\{2^{\ceil{\log_{2} n}+1}, (n+1)e\Big\},
\end{equation*}
and for that $t$, given $\phi$ and~$k$, let $q_3$, $W_A$, $W_C$, and $W_M$ be defined by Set~3 in Definition~\ref{def:prereduction2}.

From the given instance of \textsc{E-MinSAT} we contruct one and the same instance of our two control problems
as follows: Let $k$ be the limit for the number of players that can
be added, let $M$ be the set of $2k$ players that can be added with the list of
weights~$W_M$.  Further, define the quota of the WVG $\mathcal{G}$ by
\begin{align*}
	q = ~ 2\cdot \Big( & w_A + w_M + w_C  \\ 
	&  + 9\cdot\left(\sum_{i=1}^{k}10^{t(m+1)+2i}\right)  + 10^t \Big) +1,
\end{align*}
\OMIT{
	\begin{align*}
		q = & ~ 2\cdot \left(\sum_{i=1}^{n}(a_i+b_i) +  
		9\cdot\left(\sum_{i=1}^{k}10^{t(m+1)+2i}\right)\right.\\
		& + \left. \left(\sum_{j=1}^{m}\sum_{i=0}^{r}c_{j,i} \right)
		+ 10^t
		+ 1\right) +1,
	\end{align*}
} %
and let $N $ be the set of 
$14n+4nk-3k^2 -3k +m(r+1)+1$ players in
$\mathcal{G}$, subdivided into the following $14$ groups with the
following weights:  
\begin{itemize}
	\item player~$1$ with weight $1$ will be our distinguished
	player,
	
	\item group $A$ contains $2(n-k)$ players with weight list~$W_A$, 
	
	\item group $C$ contains $m(r+1)$ players with weight list~$W_C$, 
	
	\item group $D$ contains $k$ players with weight list
	\begin{align*}
		W_{D} = & ~\left(q-10^{t(m+1)+2}-w_C-2,\ldots,\right.\\
		& \hspace*{1.0mm}\left. q-10^{t(m+1)+2k}-w_C-2\right),
	\end{align*}
	\OMIT{
		\begin{align*}
			W_{D} = & ~\left(q-10^{t(m+1)+2}-\left(2^{\ceil{\log_{2} n}}-1\right)\left(\sum_{j=1}^{m} 10^{tj}\right)-2,\ldots,\right.\\
			& \hspace*{1.0mm}\left. q-10^{t(m+1)+2k}-\left(2^{\ceil{\log_{2} n}}-1\right)\left(\sum_{j=1}^{m} 10^{tj}\right)-2\right),
		\end{align*}
	} %
	
	\item group $E$ contains $n$ players, each with weight~$e$,
	
	\item group $F$ contains $n+1$ players with weights 
	$W_{F} = \left( q-q_3-1,q-q_3-e-1,\ldots,q-q_3-ne-1\right)$,
	
	\item group $G$ contains
	$(k+1)(2n-2k-1)$ players whose weights are
	of the form
	\begin{align*}
		q & -10^{t(m+1)+2}-\cdots-10^{t(m+1)+2k}-i_1 h' \\
		& -\left(2^{\ceil{\log_{2} n}}-1\right)\left(\sum_{j=1}^{m} 10^{tj}\right)
		-k-i_2 h-1
	\end{align*}
	for $i_1 \in \{0,\ldots,k\}$ and
	$i_2 \in \{1,\ldots,2n-2k-1\}$,
	
	\item group $H$ contains $2n-2k-1$ players, each with weight~$h$,
	
	\item group $H'$ contains $k$ players, each with weight~$h'$,
	
	\item group $U$ contains $2n-k+1$ players whose weights are of the form 
	\begin{align}
		q & -4\cdot 10^{t(m+1)+2}-\cdots-4\cdot 10^{t(m+1)+2k} \nonumber \\
		& - \left(2^{\ceil{\log_{2} n}}-1\right)\left(\sum_{j=1}^{m} 10^{tj}\right)-(2k+1)-iz-1
		\label{equ:case4}
	\end{align}  
	for $i\in\{0,\ldots,2n-k\}$,
	
	\item group $V$
	contains $4k-1$ players with weight list
	\begin{align*}
		W_V=\Big(
		& k+2,k+3,\ldots,2k, \\
		& 4\cdot 10^{t(m+1)+2}, \ldots, 4\cdot 10^{t(m+1)+2k}, \\
		& 3\cdot 10^{t(m+1)+2}, \ldots, 3\cdot 10^{t(m+1)+2k}, \\
		& 2\cdot 10^{t(m+1)+2}, \ldots, 2\cdot 10^{t(m+1)+2k}\Big), 
	\end{align*}
	
	\item group $X$ contains $k(2n-k+1)$ players whose weights are
	of the form
	\begin{align*}
		\qquad
		q & -5\cdot 10^{t(m+1)+2i_1}
		-w_C-2-i_2 z-1
	\end{align*}
	\OMIT{
		\begin{align*}
			\qquad
			q & -5\cdot 10^{t(m+1)+2i_1}
			-\left(2^{\ceil{\log_{2} n}}-1\right)\left(\sum_{j=1}^{m} 10^{tj}\right)-2-i_2 z-1
		\end{align*}
	} %
	for $i_1 \in\{1,\ldots,k\}$ and $i_2 \in \{0,\ldots,2n-k\}$,
	
	\item group $Y$ contains $2n-k+1$ players with weight list
	$W_Y = \left(q-1, q-z-1, \ldots, q-(2n-k)z-1\right)$, and
	
	\item group $Z$ contains $2n-k$ players, each with weight~$z$.
\end{itemize}

Let us analyze for which coalitions player~$1$ can be pivotal,
i.e., which coalitions
of $(N \cup M) \setminus \{1\}$ can have a total
weight of $q-1$. First, note that any two players from $D \cup F
\cup G \cup U \cup X \cup Y$ together have a weight larger
than~$q$; therefore, there can be at most one player from this set
in any coalition of $(N \cup M) \setminus \{1\}$ for which $1$ can be
pivotal.  Moreover, all other players together have a total weight
smaller than $q-1$. Thus any coalition
$S \subseteq (N \cup M) \setminus \{1\}$ with weight $q-1$ has to contain
exactly one of the players from $D \cup F \cup G \cup U \cup X \cup
Y$, and which other players can take part in forming such a
coalition $S$ depends on which of these groups this player belongs
to.  Accordingly, we distinguish the following six cases:
\begin{description}
	\item[Case 1:] If $S$ contains a player from~$D$ (i.e., a player
	whose weight is of the form
	\[
	q-10^{t(m+1)+2i}-\left(2^{\ceil{\log_{2} n}}-1\right)\left(\sum_{j=1}^{m} 10^{tj}\right)-2,
	\]
	for some
	$i\in\{1,\ldots,k\}$), then $S$ also has to contain exactly one
	player added from $M$ and players from some subset of~$C$, which
	is uniquely determined for each 
	player from $M$.
	
	\item[Case 2:] If $S$ contains a player from~$F$, i.e., a player
	whose weight is of the form 
	\[
	q-q_3-je-1 
	\]
	for some
	$j\in\{0,\ldots,n\}$, then $S$ also has to contain the players
	from $A \cup C \cup M$ whose weights sum up to $q_2$ and some $j$
	players with weight~$e$. 
	\OMIT{
		; the players of weights~$h$, $h'$, and $z$ and the first $k-1$ players from $V$
		have a total weight smaller than $e$ and, therefore, also smaller
		than the weight of any player in  
		$A \cup C \cup M$
	} %
	
	\item[Case 3:] If $S$ contains a player from~$G$, 
	i.e., a player whose weight is of the form
	\[
	q-\sum_{i=1}^k 10^{t(m+1)+2i} - w_C - i_1 h' - i_2 h -1
	\] 
	\OMIT{
		\[
		q-\sum_{i=1}^k 10^{t(m+1)+2i} - \left(2^{\ceil{\log_{2} n}}-1\right)\left(\sum_{j=1}^{m} 10^{tj}\right) - i_1 h' - i_2 h -1
		\] 
	} %
	for $i_1\in\{0,\ldots,k\}$ and $i_2\in\{1,\ldots,2n-2k-1\}$, then $S$ also has
	to contain $k$ players from~$M$, some players from~$C$, $H$, and
	possibly from~$H'$.
	
	\item[Case 4:] If $S$ contains a player from~$U$, 
	i.e., a player whose weight is of the form 
	\[
	q-4\sum_{i=1}^k 10^{t(m+1)+2i} - w_C - (2k+1) - i' z -1
	\]
	\OMIT{
		\[
		q-4\sum_{i=1}^k 10^{t(m+1)+2i} - \left(2^{\ceil{\log_{2} n}}-1\right)\left(\sum_{j=1}^{m} 10^{tj}\right) - (2k+1) - i' z -1
		\]
	} %
	for $i'\in\{0,\ldots,2n-k\}$, then $S$ also has
	to contain some players from $V$ and~$C$,
	at least one but at most $k-1$ players
	added from~$M$,   
	and possibly some players from~$Z$.
	
	\item[Case 5:] If $S$ contains a player from~$X$,
	i.e., a player whose weight is of the form
	\[
	q-5\cdot 10^{t(m+1)+2i}  - \left(2^{\ceil{\log_{2} n}}-1\right)\left(\sum_{j=1}^{m} 10^{tj}\right) - j z - 3
	\]
	for $i\in\{1,\ldots,k\}$ and $j\in\{0,\ldots,2n-k\}$,
	then $S$ also has
	to contain 
	the 
	pair of players of weights $a_i$ and $b_i$, 
	the 
	player from $V$ having a weight of   
	$3\cdot 10^{t(m+1)+2i}$, and     
	possibly some players from~$Z$.
	
	\item[Case 6:] If $S$ contains a player from~$Y$, i.e., a player
	whose weight is of the form 
	\[
	q-jz-1 
	\]
	for some
	$j\in\{0,\ldots,2n-k\}$, then $S$ either already achieves the
	targeted weight (namely, in case $S$ contains the player with
	weight $q-1$ from~$Y$, for $j=0$), or (if $j>0$) $S$ also has to
	contain $j$ of the players from~$Z$.
\end{description}

Note that, by the definition of these weight values, there exist no
other (than those listed above) combinations of players who could
form coalitions for which player~$1$ would be pivotal.  For example,
in Case~4, 
all other players except the player from $U$ have to have a total weight of~$4\sum_{i=1}^k 10^{t(m+1)+2i} + w_C + (2k+1) + i' z$ (for $i'$ defined above in this case).	
Each player in $A$ has too large a weight to form such a
coalition $S$ with 
the
player from $U$ (their total weight would be
greater than $q-1$). All of the players with a weight smaller than
any of those in $M$ 
have a total weight smaller than $10^{t(m+1)+1}$;
therefore, the players in 
$M\cup U \cup V$ are needed:
Each missing part of the form~$4\cdot 10^{t(m+1)+2i}$ can be achieved only by players with weights $a_i$, $b_i$, $2\cdot 10^{t(m+1)+2i}$, $3\cdot 10^{t(m+1)+2i}$, or $4\cdot 10^{t(m+1)+2i}$ because all players with smaller weights together have a total weight smaller than $10^{t(m+1)+2i}$ while 
the value~$10^{t(m+1)+2i+2}$ is too large for this part and all smaller missing parts together 
(which was shown in the proof of Lemma~\ref{lem:correspondence-assignments-weights}).  
The fact that there
has to be at least one player from $M$ is enforced by the
``$-(2k+1)$'' part
in~(\ref{equ:case4}): The smallest weights
of players in $(N \cup M) \setminus \{1\}$ are $k+2,\ldots,2k$, and to get
exactly the weight $2k+1$
(to compensate for the missing weight of $-(2k+1)$
in~(\ref{equ:case4})), we indeed need a player from~$M$.  
Analogously, $S$ cannot contain more than $k-1$ players from $M$, since it would give the nearest possible value $2k+2>2k+1$.	
Finally, by
the same argumentation, we cannot replace any of the players from $C$
or~$Z$, since any player with a larger weight alone is heavier than
all players in $C\cup Z$ together, and all players with a smaller
weight than any player from $C$ (respectively, from~$Z$) together are
lighter than that player.  The situation and the argumentation in all
other cases is analogous.

Since there are no players with weights $a_i$ or $b_i$ for
$i\in\{1,\ldots,k\}$ in the game~$\mathcal{G}$, player~$1$ can be
pivotal only for the coalitions described in 
Case~6, and therefore,
\[
\PenroseBanzhaf(\mathcal{G},1) =
\frac{\sum_{i=0}^{2n-k}{2n-k \choose i}}{2^{\|N\|-1}} =
\frac{2^{2n-k}}{2^{\|N\|-1}}.
\]

To show the correctness of the presented reductions (which obviously
can be computed in polynomial time), we need to show that the
following three statements are pairwise equivalent:
\begin{enumerate}
	\item[(1)] $(\phi,k)$ is a yes-instance of \textsc{E-MinSAT}.
	\item[(2)] $(\mathcal{G},M,1,k)$ is a yes-instance of
	\textsc{Control-by-Add\-ing-Players-to-Decrease-$\PenroseBanzhaf$}.
	\item[(3)] $(\mathcal{G},M,1,k)$ is a yes-instance of
	\textsc{Control-by-Add\-ing-Players-to-Nonincrease-$\PenroseBanzhaf$}.
\end{enumerate}

{\bf (1) $\mathbf{\Rightarrow}$ (2):}
Let us assume that $(\phi,k) \in \textsc{E-MinSAT}$, i.e.,
there exists a truth assignment to $x_1,\ldots,x_k$ such that
at most half of the assignments to the remaining $n-k$ variables yields
a satisfying assignment for the boolean formula~$\phi$.
By Lemma~\ref{lem:correspondence-assignments-weights}, these assignments correspond uniqely to subsets of $M\cup A\cup C$ with total weight~$q_2$, whereas a partial assignment to the first $k$ variables corresponds to a subset $M'$ of $M$ with $\|M'\|=k$; this is the set of new players that are added to~$\mathcal{G}$, creating a new game $\mathcal{G}_{\cup M'}$.
Therefore, there are at most $2^{n-k-1}$ assignments to $x_{n-k},\ldots,x_n$ which, together with the truth assignment to $x_1,\ldots,x_k$, satisfy~$\phi$, so there are at most $2^{n-k-1}$ subsets of players in $A\cup C \cup M'$ with total weight~$q_3$.
Now, with the players from $E \cup F$, each of these subsets can form $2^{n}$ coalitions for which player~$1$ is pivotal in $\mathcal{G}_{\cup M'}$. Moreover, $1$ becomes also pivotal for coalitions with players from $G\cup M' \cup H \cup H'$, for coalitions with players from 
$C \cup M'\cup U \cup V \cup Z$, and for $k$ coalitions with players from $M' \cup D\cup C$.  Therefore, we have
\begin{align*}
	& \PenroseBanzhaf( \mathcal{G}_{\cup M'},1) \\
	& \le \frac{2^{2n-k}+2^{2n-k}\sum_{i=1}^{k-1}{k \choose i} + 2^{n}\cdot 2^{n-k-1}}{2^{\|N\|+k-1}}  \\
	& \ \ \ \ \ \ + \frac{\sum_{i=0}^{k}{k \choose i}\sum_{j=1}^{2n-2k-1}{2n-2k-1 \choose j} + k}{2^{\|N\|+k-1}}  \\ 
	& =\frac{2^{2n-k} + 2^{2n-k}(2^k -2) + 2^{2n-k-1}+2^{k}(2^{2n-2k-1}-1)+k}{2^{\|N\|+k-1}} \\ 
	& =\frac{2^{2n-k} + 2^{k}2^{2n-k} - 2^{2n-k+1} + 2^{2n-k-1} + 2^{2n-k-1} - 2^{k} + k}{2^{\|N\|+k-1}} \\ 
	& =\frac{2^{2n-k}}{2^{\|N\|-1}} + \frac{2^{2n-k} - 2^{2n-k+1} + 2^{2n-k}-2^{k}+k}{2^{\|N\|+k-1}} \\ 
	& < \frac{2^{2n-k}}{2^{\|N\|-1}}  =  \PenroseBanzhaf(\mathcal{G},1),
\end{align*}
which means that the new Penrose--Banzhaf index of player~$1$ is
stricly smaller than the old one, so $(\mathcal{G},M,1,k)$ is a
yes-instance of
\textsc{Control-by-Adding-Players-to-Decrease-$\PenroseBanzhaf$}.

{\bf (2) $\mathbf{\Rightarrow}$ (3):} is trivially true.

{\bf (3) $\mathbf{\Rightarrow}$ (1):}
We show the contrapositive: If $(\phi,k)$ is a no-instance of
\textsc{E-MinSAT} then $(\mathcal{G},M,1,k)$ is a no-instance of
\textsc{Control-by-Adding-Players-to-Nonincrease-$\PenroseBanzhaf$}.

Let us assume now that $(\phi,k)\notin \textsc{E-MinSAT}$.
This means that there does not exist any truth assignment to
$x_1,\ldots,x_k$ such that
at most half of the assignments to the remaining $n-k$ variables yields
a satisfying assignment for the boolean formula~$\phi$, i.e., for each assignment to $x_1,\ldots,x_k$, there exist at least $2^{n-k-1}+1$ assignments to $x_{k+1},\ldots,x_n$ which satisfy~$\phi$. Let us consider possible sets of new players $M'\subseteq M$, creating after adding them the new game~$\mathcal{G}_{\cup M'}$:  
\begin{description}
	\item[Case 1:] If $\|M'\|<k$,
	then there exists some $i\in\{1,\ldots,k\}$ such that the new game
	$\mathcal{G}_{\cup M'}$
	contains none of the players with weights $a_i$ and~$b_i$, 
	so there is no coalition of weight~$q-1$ formed by players from $G\cup M'\cup C \cup H \cup H'$ and
	it is impossible to find a subset of players
	with a total weight of $q_3$ and, therefore, there is no new
	coalition for which player~$1$ can be pivotal with players from $E
	\cup F$.   
	However, there are still new coalitions for which player~$1$ can
	be pivotal, namely for each nonempty
	subset of $M'$ with players from   
	$U \cup C \cup V \cup Z$, with some players from
	$D\cup C$, and possibly with players  
	from~$X \cup C \cup Z$.
	Hence, for $k'=\|M'\|$, we have
	\begin{align*}
		\PenroseBanzhaf(\mathcal{G}_{\cup M'},1) & \ge \frac{2^{2n-k}+2^{2n-k}\sum_{i=1}^{k'}{k' \choose i}+k'}{2^{\|N\|+k'-1}}  \\ 
		& =\frac{2^{2n-k} + 2^{2n-k}(2^{k'} -1)+k'}{2^{\|N\|+k'-1}} \\ 
		& =\frac{2^{2n-k} + 2^{2n-k+k'} - 2^{2n-k} +k'}{2^{\|N\|+k'-1}} \\ 
		& >\frac{2^{2n-k}}{2^{\|N\|-1}} =\PenroseBanzhaf(\mathcal{G},1).
	\end{align*}
	
	\item[Case 2:] If $\|M'\|=k$ and $M'$ contains both players with
	weights $a_j$ and $b_j$ for some $j\in\{1,\ldots,k\}$, then
	player~$1$ is pivotal for coalitions analogously as in the
	previous case, but now we know that there are at least $2^{2n-k}$
	new coalitions with~$a_j$, $b_j$, and players from   
	$X\cup C \cup Z$, so
	\begin{align*}
		\PenroseBanzhaf( & \mathcal{G}_{\cup M'},1)\\
		& \ge \frac{2^{2n-k}+2^{2n-k}\sum_{i=1}^{k-1}{k \choose i} + 2^{2n-k}+k}{2^{\|N\|+k-1}}  \\ 
		& =\frac{2^{2n-k} + 2^{2n-k}(2^{k} -2) + 2^{2n-k}+k}{2^{\|N\|+k-1}} \\ 
		& =\frac{2^{2n-k} + 2^k 2^{2n-k} - 2\cdot 2^{2n-k} + 2^{2n-k}+k}{2^{\|N\|+k-1}} \\ 
		& >\frac{2^{2n-k}}{2^{\|N\|-1}} = \PenroseBanzhaf(\mathcal{G},1).
	\end{align*}
	
	\item[Case 3:] If $\|M'\|=k$ and $M'$ contains exactly one player with weight of each pair $\{a_i,b_i\}$ for $i\in\{1,\ldots,k\}$, then analogously to the previous implication,
	\begin{align*}
		& \quad\ \ \PenroseBanzhaf( \mathcal{G}_{\cup M'},1) \\
		& \ge \frac{2^{2n-k}+2^{2n-k}\sum_{i=1}^{k-1}{k \choose i} + 2^{n}(2^{n-k-1}+1)}{2^{\|N\|+k-1}}  \\ 
		& \ \ \ \ \ \ + \frac{ \sum_{i=0}^{k}{k \choose i}\sum_{j=1}^{2n-2k-1}{2n-2k-1 \choose j} +k}{2^{\|N\|+k-1}}  \\ 
		& =\frac{2^{2n-k}}{2^{\|N\|-1}} \\
		& + \frac{2^{2n-k} - 2^{2n-k+1} + 2^{2n-k-1}+2^{n}+2^{2n-k-1}-2^{k}+k}{2^{\|N\|+k-1}} \\
		& > \frac{2^{2n-2k}}{2^{\|N\|-1}}  =  \PenroseBanzhaf(\mathcal{G},1).
	\end{align*}
\end{description}
That means that if $(\phi,k)\notin\textsc{E-MinSAT}$, then the
Penrose--Banzhaf index of player~$1$ increases, so
$(\mathcal{G},M,1,k)$ is a no-instance of
\textsc{Control-by-Adding-Players-to-Non\-increase-$\PenroseBanzhaf$}.

Now, let $\gamma=\ShapleyShubik$.  We will prove $\np^{\PP}$-hardness
of both control
problems, \textsc{Control-by-Adding-Players-to-Decrease-$\ShapleyShubik$}
and \textsc{Control-by-Adding-Players-to-Nonincrease-$\ShapleyShubik$},
again using one and the same reduction from the $\np^{\PP}$-complete
problem $\textsc{E-MinSAT}$.

Let $(\phi, k)$ be a given instance of $\textsc{E-MinSAT}$, where
$\phi$ is a boolean formula in CNF with variables $x_1,\dots,x_n$ and
$m$ clauses, and let $k\ge 3$.

For $r= \ceil{\log_2 n}-1$, let
\begin{align*}
	P' = & ~ 8nk^3 - 4k^4 + 4n^2 k + 12nk^2 - 6k^3 \\
	& + 53n^2  -7nk + 12k^2  + 10n +11k \\
	& + (2n-k+2)m(r+1) -1, \\
	\delta = & \left\lceil\frac{5}{4}P' - 9n - \frac{9}{4}m(r+1)\right\rceil,
\end{align*}
and for these values, let
$$P=P' + \delta,$$
which is the number of players in our game.
Let
$$s=4n+m(r+1)+\delta$$
be the size of all coalitions that will be relevant in our proof, i.e., which will be counted for computing the Shapley--Shubik indices,
and let
$$k'=\frac{(P+1)\cdots (P+k)}{(P-s)\cdots (P+k-1-s)}.$$

We will now show the following bounds for $k'$ that will be used
later on in our proof:
\begin{equation}
	\label{eq:upper-lower-bound-k-prime}
	9\cdot 2^{k-3}<k'< 2^{2k}.
\end{equation}
Indeed, for some $\varepsilon \ge 0$, we have
\begin{align*}
	s = & ~ 4n + m(r+1) + \frac{5}{4}P' - 9n - \frac{9}{4}m(r+1) + \varepsilon \\
	= & ~ \frac{5}{4}P' - 5n - \frac{5}{4}m(r+1) + \frac{5}{4}\delta - \frac{5}{4}\delta + \varepsilon \\
	= & ~ \frac{4}{9}\frac{5}{4}P + \frac{5}{9}\frac{5}{4}P- 5n - \frac{5}{4}m(r+1) - \frac{5}{4}\delta + \varepsilon \\
	= & ~ \frac{5}{9}P + \frac{5}{4}\Big(\frac{5}{9}P- 4n - m(r+1) - \delta + \frac{4}{5}\varepsilon\Big)\\
	= & ~ \frac{5}{9}P + \frac{5}{4}\Big(\frac{5}{9}P'- 4n - m(r+1) - \frac{4}{9}\delta 
	+ \frac{4}{5}\varepsilon\Big)\\
	= & ~ \frac{5}{9}P + \frac{5}{4}\Big(\frac{5}{9}P'- 4n - m(r+1) - \frac{5}{9}P' \\
	& + 4n + m(r+1)-\frac{4}{9}\varepsilon 
	+ \frac{4}{5}\varepsilon\Big)\\
	\ge & ~ \frac{5}{9}P,
\end{align*}
which gives us
$$\frac{P+1}{P-s}\ge \frac{P+1}{\frac{4}{9}P} = \frac{9}{4}\Big(1+\frac{1}{P}\Big)\ge \frac{9}{4},$$
and since clearly $P\ge 9k$, it follows that
$$\frac{P+k}{P+k-1-s}\ge\frac{P+k}{\frac{4}{9}P+k-1}= 1+\frac{\frac{4}{9}P + \frac{1}{9}P+1}{\frac{4}{9}P+k-1}>2,$$
and therefore, $k' > 9\cdot2^{k-3}$, which gives the lower bound of $k'$ stated in~(\ref{eq:upper-lower-bound-k-prime}).

We prove the upper bound of $k'$ stated in~(\ref{eq:upper-lower-bound-k-prime}) as follows: 
\begin{align*}
	s = & ~ 4n + m(r+1) + \frac{5}{4}P' - 9n - \frac{9}{4}m(r+1) + \varepsilon \\
	= & ~ \frac{2}{3}P' + \frac{7}{12}P' - 5n -\frac{5}{4}m(r+1)+\varepsilon \\
	= & ~ \frac{2}{3}(P' + \delta) + \frac{7}{12}P' -\frac{2}{3}\delta - 5n -\frac{5}{4}m(r+1)+\varepsilon \\
	= & ~ \frac{2}{3}P + \frac{7}{12}P' -\frac{10}{12}P'  + 6n + \frac{3}{2}m(r+1) - \frac{2}{3}\varepsilon \\
	& - 5n -\frac{5}{4}m(r+1)+\varepsilon \\
	= & ~ \frac{2}{3}P -\frac{1}{4}P' + n + \frac{1}{4}m(r+1)  +\frac{1}{3}\varepsilon \\
	\le & ~ \frac{2}{3}P - 1,
\end{align*}
and therefore,
$$\frac{P+1}{P-s}\le \frac{P+1}{\frac{1}{3}P+1} =3\frac{P+1}{P+3}< 3,$$
which gives us $k'<3^k < 2^{2k}$, as desired.

Next, let
\begin{align*}
	y &= k' - 2^{k} > 1,\\
	\gamma_1 &= \ceil{k'-1} = k' - \varepsilon_1
	\quad\text{ for $\varepsilon_1>0$},\\
	\gamma_2 &= \floor{2^{2n-2k-1}y+1} = 2^{2n-2k-1}y + \varepsilon_2
	\quad\text{ for $\varepsilon_2>0$},\\
	\gamma_3 &= \ceil{2^{2n-2k-1}\varepsilon_1} = 2^{2n-2k-1}\varepsilon_1 + \varepsilon_3 
	\quad\text{ for $\varepsilon_3 \in [0,1)$},\\
	\gamma_4 &= \left\lceil 2^{n-k+1}y+2^n - \frac{k\varepsilon_2 + (2^k -k -2)\varepsilon_3}{2^{n-k-1}} -1 \right\rceil \\
	&= 2^{n-k+1}y+2^n - \frac{k\varepsilon_2 + (2^k -k -2)\varepsilon_3}{2^{n-k-1}} - \varepsilon_4 \\
	& \hspace*{5.4cm}\text{ for $\varepsilon_4>0$, and}\\
	\gamma_5 &= \ceil{y}.
\end{align*}

Now, for each $\gamma_i$, $i\in\{1,2,3,4,5\}$, defined above, let $\beta_i, \alpha_{i,1},\ldots,\alpha_{i,\beta_i} \in
\mathbb{N}$ be such that $\alpha_{i,1}>\cdots>\alpha_{i,\beta_i}$ and
$$\gamma_i = 2^{\alpha_{i,1}} + \cdots + 2^{\alpha_{i,\beta_i}}.$$
From the upper bound of the value $k'$ stated in~(\ref{eq:upper-lower-bound-k-prime}), we have that
\begin{align*}
	\alpha_{1,1}, \beta_1 & < 2k,\\ 
	y & = k' -2^k < 2^{2k} - 2^k,\quad\text{ and}\\
	\alpha_{5,1},\beta_5 & \le 2k.
\end{align*}
Next,
\begin{align*}
	\alpha_{2,1},\beta_2 & \le 2n,\\
	\alpha_{3,1},\beta_3 & \le 2n-2k-1,\quad\text{ and}\\
	\alpha_{4,1},\beta_4 & \le n+k+2.
\end{align*}

Now we are ready to define the groups of players, subdivided into categories, with their numbers and weights in Table~\ref{tab:adding-decrease-nonincrease-SSI}.

\onecolumn
\begin{center}
	\renewcommand{\arraystretch}{1.5}
	\begin{longtable}{c|c|c|c}
	\caption{\label{tab:adding-decrease-nonincrease-SSI}
		Groups of players in the proof of
		Theorem~\ref{restofresults}(b) and (c), with their categories,
		numbers, and weights
	} \\
		\toprule
		\textbf{Category} & \textbf{Group} & \textbf{Number of Players} & \textbf{Weights} \\
		\midrule
		& distinguished player~$p$ & $1$ & $1$ \\
		\hline
		(ms) & $A$ & $2n-2k$ & $W_A$ \\
		\hline
		(ms) & $C$ & $m(r+1)$ & $W_C$ \\
		\hline 
		(ms) & $C'$ & $m(r+1)$ & $W_{C'}$ \\
		\hline
		(size) & $D$ & $\delta$ & $d = \frac{5}{2}k^2 - \frac{7}{2}k - 2$ \\
		\hline
		(num) & \makecell{$E_i$ for \\ $i \in \{1,\ldots,\beta_4\}$} & $\alpha_{4,i}$ & $e_i = 1 + (\delta+1) d + \sum_{j=1}^{i-1}\alpha_{4,j}e_j$ \\ 
		\hline
		(size) & \makecell{$E_{i}^{*}$ for \\ $i\in\{0,\ldots,\alpha_{4,1}\}$} & $3n-1-i$ & $e_{i}^{*}=1 +(\alpha_{4,\beta_4}+1)e_{\beta_4} + \sum_{j=0}^{i-1}(3n-1-j)e_{j}^{*}$  \\
		\hline
		(num) &
		$R$ & $\sum_{i=1}^{\beta_4}(\alpha_{4,i} + 1)$ & \makecell{$q-q_{4} - j_i e_{i} - (3n-1-j_i)e_{j_i}^{*}-\delta d -1$ \\ for $i\in\{1,\ldots,\beta_4\}$ and $j_i \in\{0,\ldots,\alpha_{4,i}\}$} \\
		\hline
		(num) & $S$ &  \makecell{$(2n-2k)$\\$\cdot\sum_{i=1}^{\beta_1}(\alpha_{1,i}+1)$} & \makecell{$q -4\cdot 10^{t(m+1)+2}-\cdots-4\cdot 10^{t(m+1)+2k}-w_C$ \\
			$ -t^{*}-(3k+1)-i_1 t'' 
			- j_{i_2} t_{i_2} $ \\ $-(4n-2k-i_1-j_{i_2})t^{**}_{i_1+j_{i_2}}-\delta d-1$ \\ 
			for $i_1\in\{0,1,\ldots,2n-2k-1\}$, $i_2 \in\{1,\ldots,\beta_1\}$ \\ and $j_{i_2}\in\{0,1,\ldots,\alpha_{1,i_2}\}$} \\
		\hline
		(num) & $S'$ & $4k-2$ & \makecell{$\Big(t^{*}+2k+3,t^{*}-t^{*}_1 + 2k+4,$ \\ $\ldots,$  $t^{*} - (k-3)t^{*}_{k-3}+3k,$ \\ $ 4\cdot 10^{t(m+1)+2}, \ldots, 4\cdot 10^{t(m+1)+2k}, $ \\ $ 3\cdot 10^{t(m+1)+2}, \ldots, 3\cdot 10^{t(m+1)+2k},$ \\ $ 2\cdot 10^{t(m+1)+2}, \ldots, 2\cdot 10^{t(m+1)+2k}\Big)$} \\
		\hline
		(num) & $T$ & $2n-2k-1$ & $t'' = (3n-\alpha_{4,1})e^{*}_{\alpha_{4,1}}$ \\
		\hline
		(num) & \makecell{$T_i$ for \\ $i\in\{1,\ldots,\beta_1\}$} & $\alpha_{1,i}$ & $t_i = 1 + (2n-2k)t'' + \sum_{j=1}^{i-1} \alpha_{1,j}t_j$ \\
		\hline
		(size) & \makecell{$T^{*}_i$ for \\ $i\in\{1,\ldots,k-1\}$} & $i$ & $t_{i}^{*} =  1 +  (\alpha_{1,\beta_1}+1)t_{\beta_1} + \sum_{j=1}^{i-1} jt^{*}_j$ \\
		\hline
		(size) & \makecell{$T^{**}_i$ for \\ $i\in\{0,\ldots, $ \\ $\max\{2n-2k-1$ \\ $ +\alpha_{1,1},\alpha_{3,1}\}\}$} & $4n-2k-i$ & $t_{i}^{**} =  1 + kt^{*}_{k-1} + \sum_{j=0}^{i-1} (4n-2-j)t^{**}_j$ \\
		\hline
		(num) & $U$ & $k\sum_{i=1}^{\beta_2}(\alpha_{2,i}+1)$ & \makecell{$q -10^{t(m+1)+2i_1}-1-w_C$ \\ 
			$-j_{i_2}u_{i_2}-(4n-2-j_{i_2})u^{*}_{j_{i_2}}-\delta d-1$
			\\ for $i_1\in\{1,\ldots,k\}$, $i_2\in\{1,\ldots,\beta_2\}$ \\ and $j_{i_2}\in\{0,1,\ldots,\alpha_{2,i_2}\}$} \\
		\hline
		(num) & \makecell{$U_i$ for \\ $i\in\{1,\ldots,\beta_2\}$} & $\alpha_{2,i}$ & \makecell{$u_i = 1 + \sum_{j=1}^{i-1} \alpha_{2,j}u_j $ \\ $+
		(4n-2k-\max\{2n-2k-1+\alpha_{1,1},\alpha_{3,1}\}+1) $ \\$\cdot t^{**}_{\max\{2n-2k-1+\alpha_{1,1},\alpha_{3,1}\}}$} \\
	\hline
	(size) & \makecell{$U_{i}^{*}$ for \\ $i\in\{0,\ldots,\alpha_{2,1}\}$} & $4n-2-i$ & $u_{i}^{*} =  1 + (\alpha_{2,\beta_2}+1)u_{\beta_2} + \sum_{j=0}^{i-1} (4n-2-j)u^{*}_j$ \\
	\hline
	(num) & $V$ & $\sum_{i=1}^{\beta_3}(\alpha_{3,i}+1)$ & \makecell{$ q -4\cdot 10^{t(m+1)+2}-\cdots-4\cdot 10^{t(m+1)+2k}-w_C$ \\ 
		$%
		-t^{*}-(3k+1) 
		- j_{i} v_{i}-(4n-2k-j_i)t^{**}_{j_i}-\delta d-1$
		\\
		for $i \in\{1,\ldots,\beta_3\}$ and $j_{i}\in\{0,1,\ldots,\alpha_{3,i}\}$} \\
	\hline
	(num) & \makecell{$V_i$  for \\ $i\in\{1\ldots,\beta_3\}$} & $\alpha_{3,i}$ & $v_i = 1 + (4n-1-\alpha_{2,1})u^{*}_{\alpha_{2,1}} + \sum_{j=1}^{i-1} \alpha_{3,j}v_j$ \\
	\hline
	(num) & $X$ & $2n-k$ & \makecell{$ q-4\cdot 10^{t(m+1)+2}-\cdots-4\cdot 10^{t(m+1)+2k}-w_C$ \\ 
		$-(k-1) - iz - (4n-2k
		-i)x^{*}_i-\delta d-1$ \\
		for $i\in\{0,1,\ldots,2n-k-1\}$} \\
	\hline
	(size) & \makecell{$X_{i}^{*}$  for \\ $i\in\{0,\ldots,2n-k-1\}$} & $4n-2k-i$ & $x_{i}^{*} = 1 + (\alpha_{3,\beta_3}+1)v_{\beta_3} + \sum_{j=0}^{i-1} (4n-2k-j)x^{*}_j$ \\
	\hline
	(num) & $Y$ & \makecell{$k(2n-k+1)$ \\ $\cdot\sum_{i=1}^{\beta_5}(\alpha_{5,i}+1)$} & \makecell{ $ q -5\cdot 10^{t(m+1)+2i_1}-2 - w_C$ \\ 
		$-i_2 y'- j_{i_3}y_{i_3} -(4n-4-i_2 - j_{i_3})y^{*}_{i_2+j_{i_3}}-\delta d-1$
		\\
		for $i_1\in\{1,\ldots,k\}$, $i_2\in\{0,1,\ldots,2n-k\}$, \\ $i_3\in\{1,\ldots,\beta_5\}$ and $j_{i_3}\in\{0,1,\ldots,\alpha_{5,i_3}\}$} \\
	\hline
	(num) & $Y'$ & $2n-k$ & $y' = (2n-k+1)x^{*}_{2n-k-1}$ \\
	\hline
	(num) & \makecell{$Y_{i}$ for \\ $i\in\{1,\ldots,\beta_5\}$} & $\alpha_{5,i}$ & $y_{i} = 1 + (2n-k+1)y' + \sum_{j=1}^{i-1} \alpha_{5,j}y_j$ \\
	\hline
	(size) & \makecell{$Y_{i}^{*}$ for \\ $i\in\{0,\ldots,2n-k+\alpha_{5,1}\}$} & $4n-4-i$ & $y_{i}^{*} = 1 + (\alpha_{5,\beta_5}+1)y_{\beta_5} + \sum_{j=0}^{i-1} (4n-4-j)y^{*}_j$ \\
	\hline
	(num) & $Z$ & $2n-k$ & $q-iz-(4n+m(r+1)-1-i)z_{i}^{*}-\delta d -1$ \\
	\hline
	(num) & $Z'$ & $2n-k-1$ & $z=(2n+k-3-\alpha_{5,1})y^{*}_{2n-k+\alpha_{5,1}}$ \\
	\hline
	(size) & \makecell{$Z_{i}^{*}$ for \\ $i\in\{0,\ldots,2n-k-1\}$} & $4n+m(r+1)-1-i$ & $z_{i}^{*}= 1 + (2n-k)z + \sum_{j=0}^{i-1} (4n+m(r+1)-1-j)z^{*}_j$ \\
	\hline
	& remaining players & remaining players & $q$ \\
	\bottomrule
\end{longtable}
\end{center}
\twocolumn

Let
\[
t^{*}=(2n+m(r+1)+k+1)z^{*}_{2n-k-1},
\]
for $z^{*}_{2n-k-1}$ as defined in Table~\ref{tab:adding-decrease-nonincrease-SSI}, and let $t, t' \in \mathbb{N}$ be such that
\begin{eqnarray*}
10^{t'} & > & \max\Big\{2^{\ceil{\log_{2} n}+1}, (k-1)(t^{*}+2k+3) \Big\} \textrm{ and }\\
10^t   & > & 10^{t'} + 2^{\ceil{\log_{2} n}+1} \sum_{i=1}^{m} 10^{it'}.
\end{eqnarray*}
For that $t'$ and $t$, given $\phi$ and~$k$, let $q_4$, $W_A$, $W_C$, and $W_M$ be defined by Set~4 in Definition~\ref{def:prereduction2}.

Now, we are ready to construct the instance of our two control problems by adding players to decrease or to nonincrease a given player's Shapley--Shubik power index as follows:
\begin{itemize}
\item Let $k$ be the limit for the number of players that can be added,
\item let $M$ be the set of $2k$ players that can be added and let $W_M$ be the list of their weights,
\item let 
\begin{align*}
	q=2\cdot & \left(\sum_{i=1}^{n}(a_i+b_i) + 9\sum_{i=1}^{k}10^{t(m+1)+2i}\right. \\
	& \left. + \sum_{j=1}^{m}\sum_{s=0}^{r}(c_{j,s}) + 10^{t} + 1\right)
\end{align*}
be the quota of~$\mathcal{G}$, and
\item let $N$ be the set of $P$ players in game~$\mathcal{G}$, divided into groups with players' weights presented in Table~\ref{tab:adding-decrease-nonincrease-SSI}.
\end{itemize}
As in the proof of Theorem~\ref{thm:adding-increase-nondecrease-SSI}, each group of players in Table~\ref{tab:adding-decrease-nonincrease-SSI} (except the distinguished player~$p$
and the remaining players with weight~$q$) belongs to some category.
Among the players with category (ms) and the players in~$M$, we will focus on those coalitions whose total weight is $q_{4}$.
The main purpose of the players from the groups marked (num) is to specify the number of coalitions for which player~$p$ can be pivotal.
The players from groups with category (size) are used to make all these coalitions of equal size
and to ensure that all players with the same weight will be part of the same coalitions.
Now, we will discuss the coalitions counted in our proof in detail.

Let us first discuss which coalitions player~$p$ can be pivotal for
in any of the games $\mathcal{G}_{\cup M'}$ for some
$M' \subseteq M$.\footnote{This also includes the case of the
unchanged game~$\mathcal{G}$ itself, namely for
$M' = \emptyset$.}
Player~$p$ is pivotal for those coalitions of players in
$(N\setminus \{p\}) \cup M'$ whose total weight is $q-1$.
First, note that any two players from $F=R \cup S \cup U \cup V \cup X \cup Y\cup Z$
together have a
weight larger than $q$.
Therefore, at most one of these players
can be in any coalition player~$p$ is pivotal for.
Moreover,
by the definition of the quota, all players from $N\setminus F$ with weights different than $q$
together have a total weight smaller than $q-1$.
That means that any coalition $K \subseteq (N\setminus \{p\}) \cup
M'$ with a total weight of $q-1$ has to contain \emph{exactly} one of
the players in $F$.
Therefore, we consider the following case distinction.

\begin{description}
\item[Case 1:] If $K$ contains a player from~$R$ (with weight, say, $q-q_{4} -j_i e_{i}-(3n-1-j_i)e^{*}_{j_i}-\delta d-1$ for some~$i$, $1 \leq i \leq \beta_4$, and some~$j_i$, $0\leq j_i \leq \alpha_{4,i}$), $K$ also has to contain those players from $M \cup A \cup C \cup C'$ whose weights sum up to $q_{4}$, $j_i$ players from~$E_i$, $3n-1-j_i$ players from $E^{*}_{j_i}$, and $\delta$ players from~$D$.

\item[Case 2:] If $K$ contains a player from~$S$, then it has to contain at least one player and at most $k-2$ players from~$M$, some players from $C \cup C'$, some players from~$S'$, $i_1$ players from~$T$, $j_{i_2}$ players from~$T_{i_2}$, $4n-2k-i_1 - j_{i_2}$ players from $T^{**}_{i_1 + j_{i_2}}$, and all players from $D$, for $i_1$, $i_2$, and $j_{i_2}$ as defined for set $S$ in Table~\ref{tab:adding-decrease-nonincrease-SSI}.

\item[Case 3:] If $K$ contains a player from~$U$, it has to contain exactly one player from~$M$, some players from $C\cup C'$, $j_{i_{2}}$ players from~$U_{i_2}$, $4n-2-j_{i_2}$ players from~$U^{*}_{j_{i_2}}$, and all players from~$D$, for $i_2$ and $j_{i_2}$ as defined for set $U$ in Table~\ref{tab:adding-decrease-nonincrease-SSI}.

\item[Case 4:] If $K$ contains a player from~$V$, then $K$ also contains at least one player but at most $k-2$ players from~$M$, some players from $C \cup C'$, some players from~$S'$, $j_i$ players from~$V_i$, $4n-2k-j_i$ players from~$T^{**}_{j_i}$, and $\delta$ players from~$D$, for $i$ and $j_i$ as defined for set $V$ in Table~\ref{tab:adding-decrease-nonincrease-SSI}.

\item[Case 5:] If $K$ contains a player from~$X$, it has to contains exactly $k-1$ players from~$M$, some players from $C \cup C'$, some players from~$S'$, $i$ players from~$Z'$, and $4n-2k-2-i$ players from $X^{*}_i$ for $i\in\{0,1,\ldots,2n-k-1\}$, and all players from~$D$.

\item[Case 6:] If $K$ contains a player from~$Y$, $K$ also contains the pair $a_{i_1}$ and $b_{i_1}$, some players from $C\cup C'$, $i_2$ players from $Y'$, $j_{i_3}$ players from~$Y_{i_3}$, $4n-4-i_2-j_{i_3}$ players from~$Y^{*}_{i_2 + j_{i_3}}$, and all players from~$D$, for~$i_1$, $i_2$, $i_3$, and $j_{i_3}$ as defined for $Y$ in Table~\ref{tab:adding-decrease-nonincrease-SSI}. 

\item[Case 7:] If $K$ contains a player from $Z$ (with weight, say, $q-iz-(4n+m(r+1)-1-i)z^{*}_i-\delta d-1$ for some~$i$, $0 \leq i \leq 2n-k-1$, and some~$j_i$, $0\leq j_i \leq \alpha_{4,i}$), $K$ also has to contain $i$ players from $Z'$, $4n+m(r+1)-1-i$ players from~$Z^{*}_{i}$, and $\delta$ players from~$D$.
\end{description} 
Note that each of the coalitions described above has the same size~$s$.
Also note that there are no other combinations of coalitions with weight $q-1$ than described in the cases above due to how the players' weights were defined.
Let us analyze shortly Case~$1$ as an example.
To get weight $q-1$, $K$ has to contain (next to some player from~$R$) players with total weight $q_{4}+j_i e_i + (3n-1-j_i)e_{j_i}^{*}+\delta d$.
The part $q_{4}$ can be achieved only by the players from $M \cup A \cup C \cup C'$, since all other players from 
$N\setminus F$ with weight not greater than $t^{*}+2k+3$ have total weight smaller than $10^{t'}$ and the rest of players from $S'$ are either to large or to small to satisfy parts of $q_{4}$ (also combined with the players from $M \cup A \cup C \cup C'$). Moreover, the players from $M \cup A \cup C \cup C'$ can satisfy only $q_{4}$-part because any possible subset of that set have its weight divisible by $10^{t'}$.
For the same reason, any player from $M \cup A \cup C \cup C'$ also has a weight too big to be a part of a combination for $j_i e_i + (3n-1-j_i)e_{j_i}^{*}+\delta d$.
The value $(3n-1-j_i)e_{j_i}^{*}$ can be achieved only by the players from $E_{j_i}^{*}$ since all players from 
$D\cup \bigcup_{l=1}^{\beta_4} E_l$ and players from $E_{l}^{*}$ with smaller weight than $e_{j_i}^{*}$ together have weight smaller than any $e_{j_i}^{*}$ and the rest of players are heavier than all players from $E_{j_i}^{*}$  with all players with smaller weights together.
The same argumentation is used for the remaining value $j_i e_i + \delta d$.

Since there are no players with weights $a_i$ or $b_i$ for
$i\in\{1,\ldots,k\}$ in game~$\mathcal{G}$, player~$p$ can be
pivotal only for the coalitions described in the last case above,
i.e., in Case~7, and therefore,
\begin{align*}
\ShapleyShubik(\mathcal{G},p) = & ~ \sum_{i=0}^{2n-k-1}{2n-k-1 \choose i}\frac{s!(P-1-s)!}{P!} \\
= & ~ 2^{2n-k-1}\frac{s!(P-1-s)!}{P!}.
\end{align*}

To prove the correctness of the reduction, we show that the following
statements are pairwise equivalent:
\begin{enumerate}
\item[(1)] $(\phi, k)$ is a yes-instance of \textsc{E-MinSAT};
\item[(2)] $(\mathcal{G}, M, p, k)$ is a yes-instance of 
\textsc{Control-by-Add\-ing-Players-to-Decrease-}$\ShapleyShubik$;
\item[(3)] $(\mathcal{G}, M, p, k)$ is a yes-instance of 
\textsc{Control-by-Add\-ing-Players-to-Non\-in\-crease-}$\ShapleyShubik$.
\end{enumerate}

{\bf (1) $\mathbf{\Rightarrow}$ (2) and (1) $\mathbf{\Rightarrow}$ (3):}
Assume that $(\phi,k)$ is a yes-instance of \textsc{E-MinSAT}.
Let $M'\subseteq M$ be the set of players corresponding to some solution of $(\phi,k)$ defined according to the proof of Lemma~\ref{lem:correspondence-assignments-weights}, and let us add these players to~$\mathcal{G}$, thus creating a new game~$\mathcal{G}_{\cup M'}$.
Then there exist at most $2^{n-k-1}$ subsets of $M' \cup A \cup C \cup C'$ with total weight~$q_{4}$.
In the new game~$\mathcal{G}_{\cup M'}$, player~$p$ is still pivotal for $2^{2n-k-1}$ coalitions from Case~7 and it becomes pivotal for
\begin{itemize}
\item at most $2^{n-k-1}(2^{\alpha_{4,1}} + \cdots + 2^{\alpha_{4,\beta_4}})$ coalitions from Case~1.
\item $2^{2n-2k-1}(2^{k}-k-2)(2^{\alpha_{1,1}} + \cdots + 2^{\alpha_{1,\beta_1}})$ coalitions from Case~2,
\item $k(2^{\alpha_{2,1}} + \cdots + 2^{\alpha_{2,\beta_2}})$ coalitions from Case~3,
\item $(2^{k}-k-2)(2^{\alpha_{3,1}} + \cdots + 2^{\alpha_{3,\beta_3}})$ coalitions from Case~4, and
\item $k2^{2n-k-1}$ coalitions from Case~5,
\end{itemize}
Therefore,  
\begin{align*}
\ShapleyShubik( & \mathcal{G}_{\cup M'},p)  \\  \le & ~\Big(2^{2n-k-1}+2^{2n-2k-1}(2^{k}-k-2)\gamma_1 + k\gamma_2  \\ 
& + (2^k-k-2)\gamma_3+k2^{2n-k-1} + 2^{n-k-1}\gamma_4 \Big) \\ 
& \cdot\frac{s!(P+k-1-s)!}{(P+k)!} \\ 
= & ~ \Big(2^{2n-k-1}+2^{2n-2k-1}(2^{k}-k-2)\ceil{k'-1} \\
& + k\floor{2^{2n-2k-1}y+1}  + (2^k-k-2)\ceil{2^{2n-2k-1}\varepsilon_1}  \\
& +k2^{2n-k-1} + 2^{n-k-1}\left\lceil 2^{n-k+1}y + 2^n \right. \\
& - \frac{k\varepsilon_2 + (2^k - k -2)\varepsilon_3}{2^{n-k-1}} - 1 \rceil \Big) \frac{1}{k'}\frac{s!(P-1-s)!}{P!} \\
= & ~ \Big(2^{2n-k-1}+2^{2n-k-1}k'-2^{2n-2k-1}(k+2)(2^k+y) \\ 
& -2^{2n-2k-1}(2^k-k-2)\varepsilon_1 + k2^{2n-2k-1}y +k\varepsilon_2 \\ 
& + 2^{2n-2k-1}(2^k-k-2)\varepsilon_1 + (2^k -k -2)\varepsilon_3 +k2^{2n-k-1} \\
& + 2^{2n-2k}y + 2^{2n-k-1} - k\varepsilon_2  - (2^k - k -2)\varepsilon_3\\
& - 2^{n-k-1}\varepsilon_4 \Big) \cdot \frac{1}{k'} \cdot \frac{s!(P-1-s)!}{P!} \\
= & ~ \ShapleyShubik(\mathcal{G},p) - 2^{n-k-1}\varepsilon_4 \cdot \frac{1}{k'} \cdot \frac{s!(P-1-s)!}{P!}  \\
< & ~ \ShapleyShubik(\mathcal{G},p),
\end{align*}
so player~$p$'s Shapley--Shubik power index is strictly smaller in the new
game~$\mathcal{G}_{\cup M'}$ than in the old game~$\mathcal{G}$,
i.e., we have constructed a yes-instance of both our control problems.

{\bf (2) $\mathbf{\Rightarrow}$ (1) and (3) $\mathbf{\Rightarrow}$ (1):}
Conversely, assume that $(\phi,k)$ is a no-instance of \textsc{E-MinSAT},
i.e., 
for each value assignment for the first $k$ variables there exist at least $2^{n-k-1}+1$ value assignments for the rest variables such that together they satisfy~$\phi$.
For the set $M'\subseteq M$ corresponding to any of the solutions, we get analogously to the other implication that  
\begin{align*}
\ShapleyShubik( & \mathcal{G}_{\cup M'},p) \\
\ge & ~ \Big(2^{2n-k-1}+2^{2n-2k-1}(2^{k}-k-2)\gamma_1 + k\gamma_2 \\ 
& + (2^k-k-2)\gamma_3 +k2^{2n-k-1} + (2^{n-k-1}+1)\gamma_4 \Big) \\
& \cdot\frac{s!(P+k-1-s)!}{(P+k)!} \\ 
= & ~ \Big(2^{2n-k-1}+2^{2n-2k-1}(2^{k}-k-2)\ceil{k'-1} \\
& + k\floor{2^{2n-2k-1}y+1}  + (2^k-k-2)\ceil{2^{2n-2k-1}\varepsilon_1} \\
& +k2^{2n-k-1}  + (2^{n-k-1}+1)\ceil{2^{n-k+1}y + 2^n \\
	& - \frac{k\varepsilon_2 + (2^k - k -2)\varepsilon_3}{2^{n-k-1}} - 1} \Big) \frac{1}{k'}\frac{s!(P-1-s)!}{P!} 
\end{align*}
\begin{align*}
= & ~ \Big(2^{2n-k-1}+2^{2n-k-1}k'-2^{2n-2k-1}(k+2)(2^k+y) \\
& -2^{2n-2k-1}(2^k-k-2)\varepsilon_1 + k2^{2n-2k-1}y +k\varepsilon_2 \\
& + 2^{2n-2k-1}(2^k-k-2)\varepsilon_1 + (2^k -k -2)\varepsilon_3 +k2^{2n-k-1} \\
& + 2^{2n-2k}y + 2^{2n-k-1} - k\varepsilon_2 - (2^k - k -2)\varepsilon_3- 2^{n-k-1}\varepsilon_4 \\
& +2^{n-k+1}y+2^n - \frac{k\varepsilon_2 + (2^k - k -2)\varepsilon_3}{2^{n-k-1}} - \varepsilon_4\Big)\\
& \cdot\frac{1}{k'}\frac{s!(P-1-s)!}{P!} \\
= & ~ \ShapleyShubik(\mathcal{G},p) + \Big(- 2^{n-k-1}\varepsilon_4 +2^{n-k+1}y+2^n \\
& - \frac{k\varepsilon_2 + (2^k - k -2)\varepsilon_3}{2^{n-k-1}} - \varepsilon_4\Big)\frac{1}{k'}\frac{s!(P-1-s)!}{P!}  \\
> & ~ \ShapleyShubik(\mathcal{G},p) + \Big(- 2^{n-k-1} +2^{n-k} + 2^{n-k}+2^n \\
& - \frac{k + 2^k - k -2}{2^{n-k-1}} - \varepsilon_4\Big)\frac{1}{k'}\frac{s!(P-1-s)!}{P!}  \\
> & ~ \ShapleyShubik(\mathcal{G},p) + \Big( 2^{n-k} + \frac{2^{2n-k-1} - k - 2^k + k +2}{2^{n-k-1}} - \varepsilon_4\Big) \\
& \cdot\frac{1}{k'}\frac{s!(P-1-s)!}{P!}  \\
\ge & ~ \ShapleyShubik(\mathcal{G},p) + \Big(2^{n-k} + \frac{2^{2n-k-1}  - 2^{n-1} +2}{2^{n-k-1}} - \varepsilon_4\Big)\frac{1}{k'} \\ 
& \cdot\frac{s!(P-1-s)!}{P!}  \\
> & ~ \ShapleyShubik(\mathcal{G},p).
\end{align*}
Next, for any $M'\subseteq M$ such that $0<\|M'\|<k-1$, let
\begin{align*}
k^{*} &= \|M'\| \quad\text{and}\\
k^{**} &= \frac{P+1}{P-s}\cdots\frac{P+k^{*}}{P+k^{*}-1-s}
\end{align*}
(note that $k'\ge 2^{k-k^{*}}k^{**}$).
Then, by Cases~$2$, $3$, $4$, $7$, and possibly~$6$, we have
\begin{align*}
\ShapleyShubik( & \mathcal{G}_{\cup M'},p) \\
\ge  & ~ \Big(2^{2n-k-1} + 2^{2n-2k-1}(2^{k^{*}}-1)\gamma_1 + k^{*}\gamma_2 + (2^{k^{*}}-1)\gamma_3\Big) \\
& \cdot\frac{s!(P+k^{*}-1-s)!}{(P+k^{*})!} \\
= & ~ \Big(2^{2n-k-1} + 2^{2n-2k+k^{*}-1}k' - 2^{2n-2k-1}k'  \\
& - 2^{2n-2k-1}(2^{*}-1)\varepsilon_1 + 2^{2n-2k-1}k^{*}y+k^{*}\varepsilon_2  \\
& + (2^{k^{*}}-1)2^{2n-2k-1}\varepsilon_1 + (2^{k^{*}}-1)\varepsilon_3 \Big) \\
& \cdot \frac{1}{k^{**}}\frac{s!(P-1-s)!}{P!} \\
\ge & ~ \Big(2^{2n-k-1} + 2^{2n-2k+k^{*}-1}2^{k-k^{*}}k^{**} - 2^{2n-2k-1}y  \\
& - 2^{2n-k-1} + 2^{2n-2k-1}k^{*}y+k^{*}\varepsilon_2  + (2^{k^{*}}-1)\varepsilon_3 \Big) \\
& \cdot\frac{1}{k^{**}}\frac{s!(P-1-s)!}{P!} 
\end{align*}
\begin{align*}
\ge & ~ \Big( 2^{2n-k-1}k^{**} +k^{*}\varepsilon_2  + (2^{k^{*}}-1)\varepsilon_3 \Big)\frac{1}{k^{**}}\frac{s!(P-1-s)!}{P!} \\
> & ~ \ShapleyShubik(\mathcal{G},p).
\end{align*}
If $\|M'\|=k-1$, let $$k^{''} = \frac{P+1}{P-s}\cdots\frac{P+k-1}{P+k-2-s}$$
and then, by all the cases except for Case~$1$, we have
\begin{align*}
\ShapleyShubik( & \mathcal{G}_{\cup M'},p) \\ \ge & ~ \Big(2^{2n-k-1} + 2^{2n-2k-1}(2^{k-1}-2)\gamma_1 + (k-1)\gamma_2 \\
& + (2^{k-1}-2)\gamma_3 + 2^{2n-k-1} \Big)\frac{s!(P+k-2-s)!}{(P+k-1)!}\\
= & ~ \Big(2^{2n-k} + 2^{2n-k-2}k'- 2^{2n-2k}k'  \\
& - 2^{2n-2k-1}(2^{k-1}-2)\varepsilon_1 + 2^{2n-2k-1}(k-1)y   \\
& + (k-1)\varepsilon_2 + 2^{2n-2k-1}(2^{k-1}-2)\varepsilon_1 + (2^{k-1}-2)\varepsilon_3\Big) \\
& \cdot\frac{1}{k''}\frac{s!(P-1-s)!}{P!}\\
\ge & ~ \Big(2^{2n-k} + 2^{2n-k-1}k''- 2^{2n-2k}y - 2^{2n-k}  + (k-1)\varepsilon_2 \\
& + 2^{2n-2k-1}(k-1)y  + (2^{k-1}-2)\varepsilon_3 \Big)\frac{1}{k''}\frac{s!(P-1-s)!}{P!}\\
= & ~ \ShapleyShubik(\mathcal{G},p) + \Big( 2^{2n-2k-1}(k-3)y + (k-1)\varepsilon_2 \\
& + (2^{k-1}-2)\varepsilon_3\Big)\frac{1}{k''}\frac{s!(P-1-s)!}{P!}\\
\ge & ~ \ShapleyShubik(\mathcal{G},p) + \Big((k-1)\varepsilon_2 + (2^{k-1}-2)\varepsilon_3\Big)\frac{1}{k''}\frac{s!(P-1-s)!}{P!}\\
> & ~ \ShapleyShubik(\mathcal{G},p).\\
\end{align*}
Finally, for the remaining possible $M'$ with $\|M'\|=k$ (i.e., $M'$ contains a pair $a_i$ and $b_i$ for some $i\in\{1,\ldots,k\}$ and $k-2$ other players from~$M$), we have 
\begin{align*}
\ShapleyShubik( & \mathcal{G}_{\cup M'},p) \\  \ge & ~ \Big(2^{2n-k-1} + 2^{2n-2k-1}(2^{k}-k-2)\gamma_1 + k\gamma_2  \\
& + (2^{k}-k-2)\gamma_3 + k2^{2n-k-1} + 2^{2n-k}\gamma_5 \Big) \\
& \cdot\frac{s!(P+k-1-s)!}{(P+k)!}\\
= & ~ \Big(2^{2n-k-1}+2^{2n-k-1}k'-2^{2n-2k-1}(k+2)(2^k+y) \\
& -2^{2n-2k-1}(2^k-k-2)\varepsilon_1 + k2^{2n-2k-1}y  \\
& +k\varepsilon_2 + 2^{2n-2k-1}(2^k-k-2)\varepsilon_1 + (2^k -k -2)\varepsilon_3 \\
& +k2^{2n-k-1}+ 2^{2n-k}\ceil{y} \Big)\frac{1}{k'}\frac{s!(P-1-s)!}{P!} 
\end{align*}
\begin{align*}
= & ~ \ShapleyShubik(\mathcal{G},p) + \Big(-2^{2n-2k-1}(k+2)y + k2^{2n-2k-1}y  \\
& +k\varepsilon_2 + (2^k -k -2)\varepsilon_3 -2^{2n-k-1}+ 2^{2n-k}\ceil{y} \Big) \\
& \cdot\frac{1}{k'}\frac{s!(P-1-s)!}{P!} \\
= & ~ \ShapleyShubik(\mathcal{G},p) + \Big(-2^{2n-2k}y +k\varepsilon_2 + (2^k -k -2)\varepsilon_3 \\
& -2^{2n-k-1}+ 2^{2n-k}\ceil{y} \Big)\frac{1}{k'}\frac{s!(P-1-s)!}{P!} \\
\ge & ~ \ShapleyShubik(\mathcal{G},p) + \Big(-2^{2n-2k}y +k\varepsilon_2 -2^{2n-k-1}+ 2^{2n-k-1}y \\
& + 2^{2n-k-1}y \Big)\frac{1}{k'}\frac{s!(P-1-s)!}{P!} \\
> & ~ \ShapleyShubik(\mathcal{G},p).\\
\end{align*}
In each case, the Shapley--Shubik index of player~$p$ has decreased by
adding players, so we have constructed a no-instance of both our
control problems.  This completes the
proof.~\eproofof{Theorem~\ref{restofresults}(b) and (c)}

\sproofof{Theorem~\ref{restofresults}(d)}
Let $\gamma=\PenroseBanzhaf$.  	We will prove $\np^{\PP}$-hardness by providing a reduction from
\textsc{E-ExaSAT}.  Let
$(\phi, k, \ell)$ be a given instance of \textsc{E-ExaSAT},
where $\phi$ is a boolean formula in CNF with
variables $x_1,\dots,x_n$ and $m$ clauses,
$1\le k\le n$, and $\ell$ is an integer.

First, we need to define some values we will use in our reduction. 
For some $h\in\mathbb{N}$, let $\ell_1,\ldots,\ell_h \in \mathbb{N}$, $\ell_1 > \cdots > \ell_h$, be such that
\[
\ell=2^{\ell_1}+\cdots + 2^{\ell_h} \le 2^n
\]
(with $h\le n$ and $\ell_1\le n$).
Moreover, let $z_1=k+1$ and for $i\in\{2,\ldots,h\}$, let
\[
z_i = k + 1 + \sum_{j=1}^{i-1} \ell_j z_j.
\]
Let $t \in \mathbb{N}$ be such that 
\begin{equation}
\label{eq:assumption-for-t}
10^t > \max\Big\{2^{\ceil{\log_{2} n}+1}, k+\sum_{j=1}^{h} \ell_j z_j\Big\},
\end{equation}
and for this~$t$, given $\phi$ and~$k$, let $q_1$, $W_A$, $W_M$, and
$W_C$ be defined as in Set~1 of Definition~\ref{def:prereduction}.

Now, we contruct from $(\phi, k, \ell)$ an instance of our control
problem, \textsc{Control-by-Adding-Players-to-Maintain-$\PenroseBanzhaf$}.
Let $k$ be the limit for the number of players that can be added,
let $M$ be the set of $2k$ players that can be added with the list of weights~$W_M$, and let
\begin{align*}
q ~=~ &2 \cdot \left(w_A + w_M  + w_C  + \left(\sum_{i=1}^{h}\ell_i z_i\right) + k + 1\right) + 1 \\
~=~ &2 \cdot \left(\sum_{i=1}^{n}(a_i+b_i) + \left(\sum_{j=1}^{m}\sum_{i=0}^{r}c_{j,i}\right) + \right.\\
& \left. k + \left(\sum_{i=1}^{h}\ell_i z_i\right) + 1\right) + 1
\end{align*}
be the quota of WVG~$\mathcal{G}$.
Further, let $N$ be the set of
\[
2n+m(r+1)+2\ell_1+\cdots+2\ell_h + h+1
\]
players in $\mathcal{G}$ with the following list of weights:
\begin{align*}
W_N=( & 1,a_{k+1},\ldots,a_n, b_{k+1},\ldots,b_n, \\
& c_{1,0},\ldots,c_{1,r},\ldots,c_{m,0},\ldots,c_{m,r}, \\ 
& q-q_1-2,\ldots, q-q_1-k -1, \underbrace{1,\ldots,1}_{k}, \\
& q-1, q-z_1-1, \ldots, q-\ell_1 z_1-1, \underbrace{z_1,\ldots,z_1}_{\ell_1},
\ldots, \\
& q-1, q-z_h-1, \ldots, q-\ell_h z_h-1, \underbrace{z_h,\ldots,z_h}_{\ell_h}),
\end{align*}
which can be subdivided into the following $2h+5$ groups:
\begin{itemize}
\item player~$1$ with weight $1$ will be our distinguished player,
\item group $A$ contains $2(n-k)$ players with weight list~$W_A$,
\item group $C$ contains $m(r+1)$ players with weight list~$W_C$,
\item group $W$ contains $k$ players whose weights are of the
form $q-q_1-j-1$ for $j\in\{1,\ldots,k\}$,
\item group $X$ contains $k$ players with weight~$1$ each,
\item for each~$i\in\{1,\ldots,h\}$, there is a group $Y_i$ that
contains the players whose weights are of the form $q-j z_i - 1$
for $j \in \{0,1,\ldots,\ell_i\}$, and
\item for each~$i\in\{1,\ldots,h\}$, there is a group $Z_i$ that
contains $\ell_i$ players with weight~$z_i$.
\end{itemize}

Player~$1$ is pivotal for the coalitions
in $(N \cup M) \setminus \{1\}$ with weight $q-1$.
First, note that any two players from $W\cup Y_1 \cup \cdots \cup Y_h$ together have a total weight larger than~$q$; therefore, there can be at most one player from this set in any coalition of $S \subseteq (N \cup M) \setminus \{1\}$ for which $1$ can be pivotal.
Moreover, all players from 
$A \cup C \cup M \cup X \cup \bigcup_{i=1}^h Z_i$ together have a total weight smaller than $q-1$ (recall the definition of~$q$).
This means that any coalition $S \subseteq (N \cup M) \setminus \{1\}$ with a total weight of $q-1$ has to contain \emph{exactly} one of
the players in $W\cup Y_1 \cup \cdots \cup Y_h$.
Now, whether this player is in $W$, $Y_1$, $\ldots$, $Y_{h-1}$, or $Y_h$
has consequences as to which other players will also be in such a
weight-$(q-1)$
coalition~$S$:

\begin{description}
\item[Case 1:] If $S$ contains a player from~$W$ (with weight, say,
$q-q_1-j-1$ for some~$j$, $1 \leq j \leq k$), $S$ also has to
contain those players from $A\cup 
C\cup
M$ whose weights sum up to $q_1$
and $j$ players from $X$ with weight~$1$, but no players from~$Z_i$, for any $i\in\{1,\ldots,h\}$.
Indeed, a player of weight $z_i > k$ is too heavy to replace the
players from~$X$, and by assumption~(\ref{eq:assumption-for-t})
for~$t$, the players from $X \cup \bigcup_{i=1}^{h} Z_i$
cannot achieve the weight of any of the players from $A\cup 
C\cup
M$, so a
total weight of $q_1$ can be achieved only by the players in
$A\cup 
C \cup
M$  
(but not $q_1+j$ because any value achieved by the players is divisible by $10^t>j$).
Also, recall that $q_1$ can be achieved only by a set
of players whose weights take exactly one of the values from
$\{a_i,b_i\}$ for each $i\in\{1,\ldots,n\}$,
so $S$ must contain exactly $n-k$ players from   
$A$ that
already are in~$\mathcal{G}$ (either $a_i$ or $b_i$, for $k+1 \leq
i \leq n$) and exactly $k$ players from $M$ (either $a_i$ or
$b_i$, for $1 \leq i \leq k$); these $k$ players must have been
added to the game, i.e., $\|M'\|=k$.

\item[Case 2:] If $S$ contains a player from some
$Y_i$ 
for any $i\in\{1,\ldots,h\}$
(with weight, say, $q-1-j z_i$
for some~$j$, $0 \leq j \leq \ell_i$), then either $S$ already
achieves the weight $q-1$ for $j=0$, or $S$ has to contain $j>0$
players from~$Z_i$. The players
from 
$X\cup \bigcup_{i'=1}^{i-1} Z_{i'}$
(assuming that a sum from $1$ to $0$ is equal to~$0$)
are not heavy enough due to 
$z_i > k + \sum_{i'=1}^{i-1}\ell_{i'}z_{i'}$ and since each player from
$A\cup C\cup M$     
and each player from $Z_{l}$, $i<l\le h$, has a weight larger than $\ell_i z_i$  
together with all other $\ell_{i'}z_{i'}$, $1\le i' \le h$, $i' \neq i$,
and all players from~$X$.
\end{description}

Since there are no players with weights $a_i$ or $b_i$ for
$i\in\{1,\ldots,k\}$ in the game $\mathcal{G}$, player~$1$ can be
pivotal only for the coalitions described in the second case above,
and therefore,
\begin{eqnarray*}
\PenroseBanzhaf(\mathcal{G},1) & = &
\frac{\sum_{i=0}^{\ell_1}{\ell_1 \choose i} + \cdots + \sum_{i=0}^{\ell_h}{\ell_h \choose i}}{2^{\|N\|-1}} \\
& = & \frac{2^{\ell_1} + \cdots + 2^{\ell_h}}{2^{\|N\|-1}} \\
& = & \frac{\ell}{2^{\|N\|-1}}.
\end{eqnarray*}

We now prove the correctness of our reduction: $(\phi,k,\ell)$ is a
yes-instance of \textsc{E-ExaSAT} if and only if
$(\mathcal{G},M,1,k)$ is a yes-instance of
\textsc{Control-by-Adding-Players-to-Maintain-$\PenroseBanzhaf$}.

\proofonlyif
Assuming that $(\phi,k,\ell)$ is a yes-instance of \textsc{E-ExaSAT},
there exists an assignment to $x_1,\ldots,x_k$ such that exactly $\ell$
of the assignments to the remaining $n-k$ variables
yields a satisfying assignment for the boolean formula~$\phi$.
Let $M'\subseteq M$ be chosen as in Lemma~\ref{lem:correspondence-assignments-weights}, $\|M'\|=k$, and let $\mathcal{G}_{\cup M'}$ be the new game after adding the players to our game $\mathcal{G}$.
Since there are exactly $\ell$ truth assignments to $x_{k+1},\ldots,x_n$ for a fixed assignment to the first $k$ variables which together satisfy $\phi$, there are exactly $\ell$ subsets
of $A \cup C \cup M'$ whose elements sum up to $q_1$.
Now, with the players from $W\cup X$, each of these subsets can form $2^{k}-1$ coalitions for which player~$1$ is pivotal in~$\mathcal{G}_{\cup M'}$.
Therefore,
\begin{eqnarray*}
\PenroseBanzhaf(\mathcal{G}_{\cup M'},1)
& = & \frac{\ell +\left(2^{k}-1\right)\ell}{2^{\|N\|+k-1}}
\\ & = &
\frac{\ell 2^{k}}{2^{\|N\|+k-1}} \\
& = & \frac{\ell}{2^{\|N\|-1}} \\
& = &
\PenroseBanzhaf(\mathcal{G},1),
\end{eqnarray*}
so
the new Penrose--Banzhaf index of player~$1$ remains unchanged.

\proofif Assume now that there does not exist any assignment to
$x_1,\ldots,x_k$ such that exactly $\ell$ assignments to the
remaining $n-k$ variables satisfy the boolean formula~$\phi$, i.e.,
for each assignment to $x_1,\ldots,x_k$, there exist either fewer or
more than $\ell$ assignments to $x_{k+1},\ldots,x_n$ such that $\phi$
is satisfied.
The only possible way to maintain the Penrose--Banzhaf
power index of player~$1$ is to add to the game the new players from
$M'\subseteq M$ that uniquely correspond to the assignments to
$x_1,\ldots,x_k$ as defined in the proof of
Lemma~\ref{lem:correspondence-assignments-weights} (recall that we assume in the problem definition that at least one player must be added).  This can be seen
as follows:
\begin{itemize}
\item If $\|M'\|<k$, there exists some $i\in\{1,\ldots,k\}$
such that the new game $\mathcal{G}_{\cup M'}$ does not
contain any player with weight $a_i$ or $b_i$, so it is
impossible to find a subset of players with weight $q_1$ and
therefore there is no new coalition for which player~$1$ can
be pivotal.

\item If $\|M'\|=k$ and $M'$ contains both players with
weights $a_j$ and $b_j$ for some $j\in\{1,\ldots,k\}$, then we
get the same situation as in the previous case, because there
has to exist some $i'\in\{1,\ldots,k\}$ such that neither the
player with weight~$a_{i'}$ nor the player with
weight~$b_{i'}$ was added.
\end{itemize}
Consequently, the Penrose--Banzhaf power index of player~$1$ decreases
when $\ell\ge 1$, because the denominator increases.

Now let $M'\subseteq M$ be any subset of players that corresponds to
some assignment to $x_1,\ldots,x_k$.  By
Lemma~\ref{lem:correspondence-assignments-weights} and 
our assumption,
there are fewer or more than $\ell$ subsets of $A\cup C \cup M'$ such that
the players' weights in each subset sum up to~$q_1$.  As in the proof
of the ``Only if'' direction, for each
$j\in\{1,\ldots,k\}$, each
of these subsets of $A\cup C \cup M'$ forms a coalition of weight $q-1$
with a player in $W$ having weight $q-q_1-(j+1)$ and $j$ players
in~$X$; and there are ${k \choose j}$ of them.  Therefore, again
recalling from Case~2 above that 
$\bigcup_{i=1}^{h} (Y_i \cup Z_i)$ already contains
$\ell$ coalitions of weight $q-1$,
either
\begin{align*}
\PenroseBanzhaf(\mathcal{G}_{\cup M'},1)
&> \frac{\ell+(2^{k}-1)\ell}{2^{\|N\|+k-1}} \\
&= \frac{\ell2^{k}}{2^{\|N\|+k-1}}
=\frac{\ell}{2^{\|N\|-1}}
=\PenroseBanzhaf(\mathcal{G},1)
\end{align*}
or
\begin{align*}
\PenroseBanzhaf(\mathcal{G}_{\cup M'},1)
&< \frac{\ell+(2^{k}-1)\ell}{2^{\|N\|+k-1}} \\
&= \frac{\ell2^{k}}{2^{\|N\|+k-1}}
=\frac{\ell}{2^{\|N\|-1}}
=\PenroseBanzhaf(\mathcal{G},1),
\end{align*}
which means that the value of the Penrose--Banzhaf index of player~$1$
has changed.

Now, let $\gamma=\ShapleyShubik$. We will again prove
$\np^{\PP}$-hardness by using a reduction from
$\textsc{E-ExaSAT}$.  Let ($\phi$,$k$,$\ell$) be an instance of
$\textsc{E-ExaSAT}$, where $\phi$ is a boolean formula in CNF with
variables $x_1,\dots,x_n$ and $m$ clauses, and $\ell\ge 1$.

First, we need to define some values we will use in our reduction. For some $h\in\mathbb{N}$, let $\ell_1,\ldots,\ell_h \in \mathbb{N}$, $\ell_1 > \cdots > \ell_h$, be such that
$$\ell=2^{\ell_1}+\cdots + 2^{\ell_h}\le 2^n$$
(so, $h, \ell_1 \le n$).
Let 
$$\alpha=n^4+ 2n^3 + 13n^2 + 8n + (3n+3)m(r+1)+2$$
with $\alpha\ge 256$ (note that then $\alpha \ge 4\log_{2}^{2}\alpha$
and this holds for $n\ge 3$), and define:
\begin{align*}
P   &= \alpha^2 - k,\\
z^{*} &= 2k\floor{\log_2 \alpha} + \ell_1,\\
s    &= n+m(r+1)+z^{*}+1, \text{for $r=\ceil{\log_2 n}-1$, and}\\
k'   &=\frac{(P+1)\cdots (P+k)}{(P-s)\cdots (P+k-1-s)}.
\end{align*}
Further, define
$$y=(P-s)\cdots (P+k-1-s)$$
and let $y_1,\ldots,y_{h'}\in\mathbb{N}$, $y_1 > \cdots > y_{h'}$, be such that
$$y=2^{y_1}+\cdots + 2^{y_{h'}},$$
and define
$$z=(P+1)\cdots (P+k)-y$$
and let $z_1,\ldots,z_{h''}\in\mathbb{N}$, $z_1 > \cdots > z_{h''}$, be such that
$$z=2^{z_1}+\cdots + 2^{z_{h''}}.$$
Note that $y,z < (P+k)^{k}$, and therefore,
\[
y_1, z_1, h', h'' < 2k\log_2 \alpha.
\]
\OMIT{
Moreover, let $u_1=1$ and for $i\in\{2,\ldots,h\}$, let
$$u_i = 1 + \sum_{j=1}^{i-1} z_j u_j,$$
$u_1 ' = (\ell_{h''}+1)u_{h''}$ and for $i \in\{2,\ldots,z_1 +1\}$, let
$$u_i ' = 1 + \sum_{j=1}^{i-1} (z^{*}-j+1) u_j ',$$
$v_1 = (z^{*}-z_1+1)u_{z_1 +1}'$ and for $i \in \{2,\ldots,h\}$, let
$$v_i = 1 + \sum_{j=1}^{i-1} \ell_j v_j,$$
$w_1 = (\ell_h + 1)v_h$ and for $i\in\{2,\ldots,h'\}$, let
$$w_i  = 1 + \sum_{j=1}^{i-1} y_j w_j ,$$
$w_1 ' = (y_{h'}+1)w_{h'}$ and for $i\in\{2,\ldots, \ell_1 + y_1 +1\}$, let
$$w_i ' = 1 + \sum_{j=1}^{i-1} (s-j) w_j '.$$
} %
Let $t' \in \mathbb{N}$ be such that 
\begin{equation*}
10^{t'} > \max\left\{2^{\ceil{\log_{2} n}+1}, (\ell_1+ y_1 + 2)w'_{\ell_1 + y_1 + 1}\right\},
\end{equation*}
for $w'_{\ell_1 + y_1 + 1}$ as defined in Table~\ref{tab:adding-maintain-SSI},
and for this~$t'$, given $\phi$ and~$k$, we define the values of~$t$, $q_2$
$W_A$, $W_M$, $W_C$, and $W_{C'}$ as in Set~2 of
Definition~\ref{def:prereduction}.

\begin{table*}
	\caption{\label{tab:adding-maintain-SSI}
		Groups of players in the proof of
		Theorem~\ref{restofresults}(d), with their categories,
		numbers, and weights
	}
\begin{center}
	\renewcommand{\arraystretch}{1.5}
	\begin{tabular}{c|c|c|c}
		\toprule
		\textbf{Category} & \textbf{Group} & \textbf{Number of Players} & \textbf{Weights} \\
		\midrule
		& distinguished player~$1$ & $1$ & $1$ \\
		\midrule
		(sat) & $A$ & $2n-2k$ & $W_A$ \\
		\midrule
		(sat) & $C$ & $m(r+1)$ & $W_C$ \\
		\midrule 
		(sat) & $C'$ & $m(r+1)$ & $W_{C'}$ \\
		\midrule
		(num) & \makecell{$T_i$ \\ for $i\in\{1,\ldots,h''\}$} & $z_i + 1$ & \makecell{$q-q_2-j u_i - (z^{*}-j)u_{j+1}'-1$ \\ for $j\in\{0,\ldots,z_i\}$} \\
		\midrule
		(num) & \makecell{$U_i$ \\ for $i\in\{1,\ldots,h''\}$} & $z_i$ &
		$u_i = 1 + \sum_{j=1}^{i-1} z_j u_j$ \\
		\midrule
		(size) & \makecell{$U_i '$ \\ for $i\in\{0,\ldots,z_1\}$} & $z^{*}-i$ &
		$u_{i+1} ' = (z_{h''}+1)u_{h''} +  \sum_{j=1}^{i-1} (z^{*}-j+1) u_j '$ \\
		\midrule
		(num) & \makecell{$V_i$ \\ for $i\in\{1,\ldots,h\}$} & $\ell_i$ & $v_i = (z^{*}-z_1+1)u_{z_1 +1}' + \sum_{j=1}^{i-1} \ell_j v_j$ \\
		\midrule
		(num) & \makecell{$W_i$ \\ for $i\in\{1,\ldots,h'\}$} & $y_i$ & $w_i = (\ell_h + 1)v_h + \sum_{j=1}^{i-1} y_j w_j$ \\
		\midrule
		(size) & \makecell{$W_i '$ \\ for $i\in\{1,\ldots,\ell_1 + y_1 + 1\}$} & $s-i$ & $w_i ' = (y_{h'}+1)w_{h'} + \sum_{j=1}^{i-1} (s-j) w_j '$ \\
		\midrule
		(num) & \makecell{$X_{i,j}$ \\ for $i \in\{1,\ldots,h\}$, $j \in\{1,\ldots,h'\}$} & 
		$(\ell_i +1)(y_j +1)$ & \makecell{$q-\ell_{i}' v_{i} - \ell_{j}' w_{j} - (s-1-\ell_{i}'-\ell_{j}') w_{\ell_i ' +\ell_j ' +1}' - 1$ \\    
			for $\ell_{i}' \in\{0,\ldots,\ell_{i}\}$, $\ell_{j}'\in\{0,\ldots,y_{j}\}$} \\
		\midrule
		& $Y$ & remaining players & $q$ \\
		\bottomrule
	\end{tabular}
\end{center}
\end{table*}

Now, we construct the instance of our control problem,
\textsc{Control-by-Adding-Players-to-Maintain-$\ShapleyShubik$}.
Let $k$ be the limit for the number of players that can be added, and
let $M$ be the set of $2k$ players that can be added with the list of
weights~$W_M$.
Let $N$ be the set of $P$ players in the  game $\mathcal{G}$, subdivided
into the following groups with their categories, numbers, and weights
as presented in Table~\ref{tab:adding-maintain-SSI}.
Among the players from $A\cup M\cup C \cup C'$, we will focus on those subsets whose total weight is~$q_2$.
The players from 
$\bigcup_{i=1}^{h''}(T_i \cup U_i)$
and from 
$\bigcup_{i=1}^{h'}\left(W_i \cup \bigcup_{j=1}^{h} X_{i,j}\right) \cup \bigcup_{i=1}^h V_i$
define the number of coalitions for which the distinguished
player~$1$ can be pivotal, and the players from sets
$U_i '$, $i\in\{0,\ldots,z_1\}$, and $W_i '$, $i\in\{1,\ldots,\ell_1 + y_1 + 1\}$, 
make all these coalitions equally large.
In the following, we will discuss these coalitions in detail.

Finally, let $q^{*}$ be the total weight of all players from 
\[
\left(N\setminus \left(\bigcup_{i=1}^{h''}T_i \cup \bigcup_{i=1}^{h}\bigcup_{j=1}^{h'}\left(X_{i,j} \cup Y\right)\right)\right)\cup M
\]
and define the quota of $\mathcal{G}$ by
\[
q=2q^{*}+1.
\]

Let us first discuss which coalitions player~$1$ can be pivotal for
in any of the games $\mathcal{G}_{\cup M'}$ for some
$M' \subseteq M$.\footnote{This also includes the case of the
	unchanged game~$\mathcal{G}$ itself, namely for
	$M' = \emptyset$.}
Player~$1$ is pivotal for those coalitions of players in
$(N\setminus \{1\}) \cup M'$ whose total weight is $q-1$.
First, note that any two players from 
$\left(\bigcup_{i=1}^{h''}T_i\right) \cup \left(\bigcup_{i=1}^{h}\bigcup_{j=1}^{h'}X_{i,j}\right)$
together have a
weight larger than~$q$.  Therefore, at most one player from 
$\left(\bigcup_{i=1}^{h''}T_i\right) \cup \left(\bigcup_{i=1}^{h}\bigcup_{j=1}^{h'}X_{i,j}\right)$
can be in any coalition player~$1$ is pivotal for.  Moreover,
by the definition of our quota, all players from 
\begin{align*}
	& A \cup C \cup C' \cup M \cup \\
	& \left(\bigcup_{i=1}^{h''}U_i\right) \cup
	\left(\bigcup_{i=0}^{z_1} U_i '\right) \cup
	\left(\bigcup_{i=1}^{h} V_i\right) \\
	& \cup
	\left(\bigcup_{i=1}^{h'}W_i\right)  \cup
	\left(\bigcup_{i=1}^{\ell_1 + y_1 + 1} W_i '\right)
\end{align*}
together have a total weight smaller than $q-1$. 
That means that any coalition $S \subseteq (N\setminus \{1\}) \cup
M'$ with a total weight of $q-1$ has to contain \emph{exactly} one of
the players in 
$\left(\bigcup_{i=1}^{h''}T_i\right) \cup
\left(\bigcup_{i=1}^{h}\bigcup_{j=1}^{h'}X_{i,j}\right)$.
Now, whether this player is in $T_i$ or in~$X_{j_1,j_2}$
for some $i$, $j_1$, and~$j_2$ with
$1\le i\le h''$, $1\le j_1\le h$, and $1\le j_2 \le h'$,
has consequences as to which other players will also be in such a
weight-$(q-1)$ coalition~$S$:

\begin{description}
	\item[Case 1:] If $S$ contains a player from~$T_i$ (with weight, say, 
	$q-q_2-j_i u_i-\left(z^{*}-j_i\right)u_{j_i +1}'-1$ 
	for some~
	$j_i$, $0 \leq j_i \leq z_i$), $S$ also has to
	contain those players from $A\cup 
	C\cup C' \cup M$ whose weights sum up to $q_2$,
	all players from $U_{j_i +1}'$,
	and $j_i$ players from 
	$U_i$
	with weight~$u_i$, but no players 
	from~$\bigcup_{i=1}^h V_i$, $\bigcup_{i=1}^h W_i$,
	or $\bigcup_{i=1}^{\ell_1 + y_1 + 1} W_i '$.
	Also, recall that $q_2$ can be achieved only by a set
	of players whose weights take exactly one of the values from
	$\{a_i,b_i\}$ for each $i\in\{1,\ldots,n\}$,
	so $S$ must contain exactly $n-k$ players from   
	$A$ that
	already are in~$\mathcal{G}$ (either $a_i$ or $b_i$, for $k+1 \leq
	i \leq n$) and exactly $k$ players from $M$ (either $a_i$ or
	$b_i$, for $1 \leq i \leq k$); these $k$ players must have been
	added to the game, i.e., $\|M'\|=k$.

	\item[Case 2:] If $S$ contains a player from $X_{j_1,j_2}$
	(with weight, say, $q-\ell_{j_1}' v_{j_1} - \ell_{j_2}' w_{j_2} - (s-1-\ell_{j_{1}}'-\ell_{j_{2}}') w_{\ell_{j_{1}}+\ell_{j_2} +1}' - 1$
	for some~$\ell_{j_1}'$, $0 \leq \ell_{j_1}' \leq \ell_{j_1}$, and some~$\ell_{j_2}'$, $0 \leq \ell_{j_2}' \leq y_{j_2}$), then either $S$ already
	achieves weight $q-1$ for $\ell_{j_1}'=\ell_{j_2}'=0$, or $S$ has to contain $\ell_{j_1}'$ players from~$V_{j_1}$ and
	$\ell_{j_2}'$
	players from~$W_{j_2}$,
	and $s-1-\ell_{j_1}'-\ell_{j_2}'$ players from~$W_{\ell_{j_1}'+\ell_{j_2}'+1}'$.
\end{description}
Note that all coalitions described above have the same size of~$s$.

Since there are no players with weights $a_i$ or $b_i$ for
$i\in\{1,\ldots,k\}$ in the game $\mathcal{G}$, player~$1$ can be
pivotal only for the coalitions described in the second case above,
and therefore,
\begin{align*}
	\ShapleyShubik(\mathcal{G},1) & =\Big(2^{y_1}+\cdots+2^{y_{h'}}\Big)\Big(2^{\ell_1}+\cdots + 2^{\ell_h}\Big)\frac{s!(P-1-s)!}{P!} \\
	& =y\cdot \ell\cdot \frac{s!(P-1-s)!}{P!}>0.
\end{align*}

We now show the correctness of our reduction: $(\phi, k, \ell)$ is a
yes-instance of $\textsc{E-ExaSAT}$ if and only if
$(\mathcal{G}, M, 1, k)$ as defined above is a yes-instance of
\textsc{Control-by-Adding-Players-to-Maintain-$\ShapleyShubik$}.

\proofonlyif
Suppose that $(\phi, k, \ell)$ is a yes-instance of $\textsc{E-ExaSAT}$,
i.e., there exists an assignment to $x_1,\ldots,x_k$ such that exactly
$\ell$ assignments to the remaining $n-k$ variables 
yields a satisfying assignment for the boolean formula~$\phi$.
Let us fix one of these
satisfying assignments.  From this fixed assignment, we get the
vector $\vec{d} = (d_1,\ldots,d_n)$ as defined in the proof of
the analogue of Lemma~\ref{lem:correspondence-assignments-weights}
for Set~2 and $q_2$
from Definition~\ref{def:prereduction2}, where the first
$k$ positions correspond to the players $M'\subseteq M$, $\|M'\|=k$,
which we add to the game~$\mathcal{G}$.

Since there are exactly $\ell$ assignments to
$x_{n-k},\ldots,x_n$ which---together with the fixed assignments to
$x_1,\ldots,x_k$---satisfy~$\phi$, by
the analogue of Lemma~\ref{lem:correspondence-assignments-weights}
for Set~2 and $q_2$
from Definition~\ref{def:prereduction2},
there are exactly $\ell$ subsets of
$A\cup C \cup C' \cup M'$
such that the players' weights
in each subset sum up 
to~$q_2$.
Each of these subsets with total weight
$q_2$
can form 
coalitions of weight $q-1$
(i.e., coalitions player~$1$ is pivotal for in the new
game~$\mathcal{G}_{\cup M'}$) with each player from 
$\bigcup_{i=1}^{h''}T_i$---and there are $2^{z_1}+\cdots +2^{z_{h''}}=z$ such coalitions.
Therefore, recalling from Case~2 above that 
player~$1$ is 
already pivotal for
$y\cdot \ell$ coalitions of weight $q-1$, we have
\begin{align*}
	\ShapleyShubik(\mathcal{G}_{\cup M'},1) 
	= & ~ \Big(y\cdot \ell + z\cdot \ell\Big)\frac{s!(P+k-1-s)!}{(P+k)!} \\
	= & ~ \Big(y\cdot \ell + ((P+1)\cdots (P+k)-y)\cdot \ell\Big) \\
	& \cdot \frac{s!(P-1-s)!}{P!}\frac{(P-s)\cdots (P+k-1-s)}{(P+1)\cdots (P+k)} \\
	= & ~ \Big(y\frac{(P+1)\cdots (P+k)}{y}\cdot \ell\Big)\frac{1}{k'}\frac{s!(P-1-s)!}{P!} \\
	= & ~ \Big(y\cdot k' \cdot \ell\Big)\frac{1}{k'}\frac{s!(P-1-s)!}{P!} = \ShapleyShubik(\mathcal{G},1), 
\end{align*}
so player~$1$'s Shapley--Shubik index remains unchanged,
i.e., we have constructed a yes-instance of our control problem.

\proofif
Assume now that $(\phi, k, \ell)$ is a no-instance of $\textsc{E-ExaSAT}$,
i.e., there does not exist any
assignment to the variables $x_1,\ldots,x_k$ such that exactly $\ell$ assignments to the remaining $n-k$ variables 
yields a satisfying assignment for the boolean
formula~$\phi$.  In other words, for each assignment to $x_1,\ldots,x_k$,
there exist either fewer or more than $\ell$ assignments to $x_{k+1},\ldots,x_n$
which satisfy~$\phi$.
Again, we consider subsets $M'\subseteq M$ of players that uniquely
correspond to the assignments to $x_1,\ldots,x_k$ according
to~(\ref{addingplayers-increasePBI:def:d_i}).
Note that
\begin{itemize}
	\item if $\|M'\|<k$, then there exists some $i\in\{1,\ldots,k\}$ such
	that the new game $\mathcal{G}_{\cup M'}$ does not contain any
	player of weight $a_i$ or~$b_i$, so it is impossible to find a
	subset of players with total weight
	$q_2$ and, therefore, there is
	no new coalition player~$1$ may be pivotal for;
	\item if $\|M'\|=k$ and $M'$ contains both the player of weight $a_j$
	and the player of weight $b_j$ for some $j \in \{1,\ldots,k\}$, then
	we get the same situation as in the previous case: There is no new
	coalition player~$1$ may be pivotal for
	because there is some $j' \in \{1,\ldots,k\}$ such that neither the player with weight~$a_{j'}$ nor the player with weight~$b_{j'}$ has been added to~$\mathcal{G}$.
\end{itemize}
In both cases above, the Shapley--Shubik index of player~$1$  decreases.

Now let $M'\subseteq M$ be any subset of players that corresponds to
some assignment to $x_1,\ldots,x_k$.  By
the analogue of Lemma~\ref{lem:correspondence-assignments-weights} for
Set~2 and~$q_2$ and our assumption,
there are either fewer or more than $\ell$ subsets of $A\cup C \cup C' \cup M'$ such that
the players' weights in each subset sum up
to~$q_2$.  As in the proof
of the ``Only if'' direction,
each of these subsets of $A\cup C \cup C' \cup M'$ forms
a coalition of weight $q-1$ with a player in $\bigcup_{i=1}^{h''}T_i$ and some players from $\left(\bigcup_{i=1}^{h''} U_i\right)\cup\left(\bigcup_{i=1}^{z_1}U_i '\right)$---and there are
$z$ of them.
Therefore, again
recalling from Case~2 above that
player~$1$ is already pivotal for
$y\cdot \ell$ coalitions, 
we have either
\begin{align*}
	\ShapleyShubik(\mathcal{G}_{\cup M'},1) &
	> \Big(y\cdot \ell + z\cdot \ell\Big)\frac{s!(P+k-1-s)!}{(P+k)!} \\
	& = \Big(y\cdot k' \cdot \ell\Big)\frac{1}{k'}\frac{s!(P-1-s)!}{P!} = \ShapleyShubik(\mathcal{G},1)
\end{align*}
or
\begin{align*}
	\ShapleyShubik(\mathcal{G}_{\cup M'},1) &
	< \Big(y\cdot \ell + z\cdot \ell\Big)\frac{s!(P+k-1-s)!}{(P+k)!} \\
	& = \Big(y\cdot k' \cdot \ell\Big)\frac{1}{k'}\frac{s!(P-1-s)!}{P!} = \ShapleyShubik(\mathcal{G},1).
\end{align*}
Thus, also in this case, the Shapley--Shubik index of player~$1$
cannot stay unchanged by adding up to $k$ players from $M$ to the
game~$\mathcal{G}$, and we have a no-instance of our control
problem.~\eproofof{Theorem~\ref{restofresults}(d)}

\end{document}